\documentclass{article}
\usepackage[margin=1in]{geometry}
\usepackage{amsmath,amsthm,amssymb,bbm}
\usepackage{url} 
\usepackage{tabularx,tabu}
\usepackage{hyperref}
\usepackage{booktabs} 
\usepackage{graphicx}
\usepackage[]{xcolor} 
\newcommand{\p}[1]{\left(#1\right)}
\renewcommand{\b}[1]{\left[#1\right]}
\newcommand{\E}[1]{\mathbb{E}\b{#1}}
\def\given{\;|\;}
\newcommand\numberthis{\addtocounter{equation}{1}\tag{\theequation}}

\def\dx{\;dx}
\def\dr{\;dr}

\newcommand{\ind}[1]{\mathbbm{1}_{#1}}
\def\R{\mathbb{R}}

\newcommand{\mc}[1]{\mathcal{#1}}

\DeclareMathOperator\erf{erf}
\DeclareMathOperator\cont{cont}
\DeclareMathOperator\swap{swap}
\DeclareMathOperator\inv{inv}

\newtheorem{theorem}{Theorem}
\newtheorem{lemma}{Lemma}
\newtheorem{definition}{Definition}
\theoremstyle{definition}

\title{Algorithmic Monoculture and Social
Welfare\thanks{A version of this paper appears in Proceedings of the National
Academy of Sciences at \url{https://www.pnas.org/content/118/22/e2018340118}}}
\author{Jon Kleinberg \and Manish Raghavan}
\date{}
\begin{document}
\maketitle
\begin{abstract}
  As algorithms are increasingly applied to screen applicants for high-stakes
  decisions in employment, lending, and other domains, concerns have been raised
  about the effects of {\em algorithmic monoculture}, in which many
  decision-makers all rely on the same algorithm. This concern invokes analogies
  to agriculture, where a monocultural system runs the risk of severe harm from
  unexpected shocks. Here we show that the dangers of algorithmic monoculture
  run much deeper, in that monocultural convergence on a single algorithm by a
  group of decision-making agents, even when the algorithm is more accurate for
  any one agent in isolation, can reduce the overall quality of the decisions
  being made by the full collection of agents. Unexpected shocks are therefore
  not needed to expose the risks of monoculture; it can hurt accuracy even under
  ``normal'' operations, and even for algorithms that are more accurate when
  used by only a single decision-maker. Our results rely on minimal assumptions,
  and involve the development of a probabilistic framework for analyzing systems
  that use multiple noisy estimates of a set of alternatives.
\end{abstract}
 \section{Introduction}

The 
rise of algorithms used to shape societal choices has been accompanied by
concerns over \textit{monoculture}---the notion that choices and preferences
will become homogeneous in the face of algorithmic curation. 
One of many canonical articulations of 
this concern was expressed in
the New York Times by Farhad Manjoo, who wrote, 
``Despite the barrage of choice, more of us are
enjoying more of the same songs, movies and TV shows''~\cite{manjoo2019this}.
Because of algorithmic curation, trained on collective social
feedback \cite{salganik2006experimental}, our choices are converging.

When we move from the influence of algorithms on media consumption
and entertainment to their influence on high-stakes 
screening decisions about whom to offer a job or whom to offer a loan,
the concerns about algorithmic monoculture become even starker.
Even if algorithms are more accurate on a case-by-case basis,
a world in which everyone uses the same algorithm is susceptible to
correlated failures when the algorithm finds itself in adverse conditions.
This type of concern invokes an analogy to agriculture, where
monoculture makes crops susceptible to the attack of a single pathogen
\cite{power1987monoculture}; the analogy has become a mainstay of
the computer security literature \cite{birman2009monoculture}, and
it has recently become a source of concern about screening decisions
for jobs or loans as well.
Discussing the post-recession
financial system, Citron and Pasquale write, ``Like monocultural-farming
technology vulnerable to one unanticipated bug, the converging methods of credit
assessment failed spectacularly when macroeconomic conditions
changed''~\cite{citron2014scored}.

The narrative around algorithmic monoculture thus suggests a trade-off:
in ``normal'' conditions, a more accurate algorithm will improve the
average quality of screening decisions, but when conditions change through
an unexpected shock, the results can be dramatically worse.
But is this trade-off genuine?  In the absence of shocks, 
does monocultural convergence on a single, more accurate screening algorithm
necessarily lead to better average outcomes?

In this work, we show that algorithmic monoculture poses risks
even in the absence of shocks.
We investigate a model involving minimal assumptions, in which two
competing firms can either use their own independent
heuristics to perform screening decisions 
or they can use
a more accurate algorithm that is accessible to both of them.
(Again, we think of screening job applicants or loan
applicants as a motivating scenario.)
We find that even though 
it would be rational for each firm in isolation to adopt the algorithm,
it is possible for the use of the algorithm by both firms
to result in decisions that are {\em worse} on average.
This in turn leads, in the language of game theory,
to a type of ``Braess' paradox'' 
\cite{braess1968paradoxon}
for screening algorithms:
the introduction of a more accurate algorithm can drive the firms into
a unique equilibrium that is worse for society than the one that was present
before the algorithm existed.

Note that the harm here is to {\em overall} performance.
Another common concern about algorithmic monoculture in screening 
decisions is the harm it can cause to specific individuals:
if all employers or lenders use the same
algorithm for their screening decisions, then particular applicants
might find themselves locked out of the market when this shared algorithm
doesn't like their application for some reason.
While this is clearly also a significant concern, 
our results show that it would be a mistake
to view the harm to particular applicants
as necessarily balanced against the gains in overall accuracy ---
rather, it is possible for algorithmic monoculture to cause harm not just
to particular applicants but also
to the {\em average} quality of decisions as well.

Our results thus have a counterintuitive flavor to them: if an
algorithm is clearly more accurate than the alternatives
when one entity uses it, why
does the accuracy become worse than the alternatives
when multiple entities use it?
The analysis relies on deriving some novel probabilistic properties of
rankings, establishing that when we are
constructing a ranking from a probability distribution representing a
``noisy'' version of a true ordering, we can sometimes achieve less error
through an incremental construction of the ranking --- building it
one element at a time --- than we can 
by constructing it in a single draw from the distribution.
We now set up the basic model, and then frame the probabilistic
questions that underpin its analysis.

\begin{figure*}[ht]
  \centering
  \includegraphics[width=0.8\linewidth]{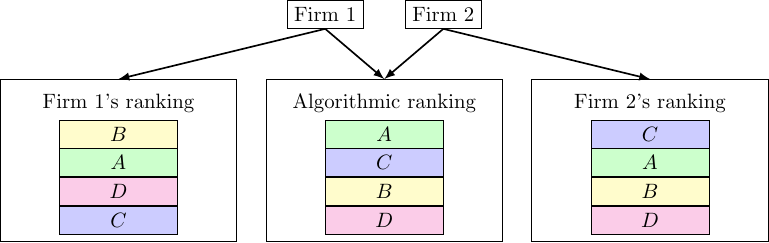}
  \caption{Each firm has the choice to use either their own private ranking or
    the common algorithmic ranking to order the $n$ candidates. In a random order,
    each firm hires the highest-ranked available candidate according to the
    ranking they chose. For example, if Firm 1 uses their private ranking and Firm
    2 uses the algorithmic ranking, then Firm 1 hires candidate $B$ and Firm 2
    hires candidate $A$. If both firms use the algorithmic ranking, then the
    firm randomly selected to hire first hires candidate $A$ and the firm
  randomly selected to hire second hires candidate $C$.}%
  \label{fig:example}
\end{figure*}

\section{Algorithmic hiring as a case study}

To instantiate the ideas introduced thus far, 
we'll focus on the case of algorithmic hiring,
where recruiters make decisions based in part on scores or recommendations
provided by data-driven algorithms. In this setting, we'll propose and analyze a
stylized model of algorithmic hiring with which we can begin to investigate the
effects of algorithmic monoculture.

Informally, we can think of a simplified hiring process as follows: rank all of
the candidates and select the first available one.
We suppose that each firm has two options to form this ranking: either develop
their own, private ranking (which we will refer to as using a ``human
evaluator''), or use an algorithmically produced ranking. We assume that there
is a single vendor of algorithmic rankings, so all firms choosing to use the
algorithm receive the same ranking. The firms proceed in a random order, each hiring
their favorite remaining candidate according to the ranking they're
using---human-generated or algorithmic (see Figure~\ref{fig:example} for an
example). Thus, we can frame the effects of
monoculture as follows: are firms better off using the more accurate, common
algorithm, or should they instead employ their own less accurate, but private,
evaluations?

In what follows, we'll introduce a formal model of evaluation and selection,
using it to analyze a setting in which firms seek to hire candidates.

\subsection{Modeling ranking}

More formally, we model the $n$ candidates as having intrinsic values $x_1,
\dots, x_n$, where any employer would derive utility $x_i$ from hiring candidate
$i$. Throughout the paper, we assume without loss of generality that $x_1 > x_2
> \dots > x_n$. These values, however, are unknown to the employer; instead,
they must use some noisy procedure to rank the candidates. We model such a
procedure as a randomized mechanism $\mc R$ that takes in the true candidate
values and draws a permutation $\pi$ over those candidates from some
distribution. Our main results hold for families of distributions over
permutations as defined below:
\begin{definition}[Noisy permutation family]
  \label{def:family}
  A noisy permutation family $\mc F_\theta$ is a family of
  distributions over permutations that satisfies the following conditions for
  any $\theta > 0$ and set of candidates $\mathbf{x}$:
  \begin{enumerate}
    \item \textbf{(Differentiability)} For any permutation $\pi$, $\Pr_{\mc
      F_\theta}[\pi]$ is continuous and differentiable in $\theta$.
    \item \textbf{(Asymptotic optimality)} For the true ranking $\pi^*$,
      $\lim_{\theta \to \infty} \Pr_{\mc F_\theta}[\pi^*] = 1$.
    \item \textbf{(Monotonicity)} For any (possibly empty) $S \subset
      \mathbf{x}$, let $\pi^{(-S)}$ be the partial ranking produced by removing
      the items in $S$ from $\pi$. Let $\pi_1^{(-S)}$ denote the value of the
      top-ranked candidate according to $\pi^{(-S)}$. For any $\theta' >
      \theta$,
      \begin{equation}
        \label{eq:monotonicity}
        \mathbb{E}_{\mc F_{\theta'}}\b{\pi_1^{(-S)}} \ge \mathbb{E}_{\mc
        F_{\theta}}\b{\pi_1^{(-S)}}.
      \end{equation}
      Moreover, for $S = \emptyset$,~\eqref{eq:monotonicity} holds with strict
      inequality.
  \end{enumerate}
\end{definition}

$\theta$ serves as an ``accuracy parameter'': for large $\theta$, the noisy
ranking converges to the true ranking over candidates. The monotonicity
condition states that a higher value of $\theta$ leads to a better first choice,
even if some of the candidates are removed after ranking. Removal after ranking
(as opposed to before) is important because some of the ranking models we will
consider later do not satisfy Independence of Irrelevant Alternatives. Examples
of noisy permutation families include Random Utility
Models~\cite{thurstone1927law} and the Mallows Model~\cite{mallows1957non}, both
of which we will discuss in detail later.

As an objective function to evaluate the effects of different approaches
to ranking and selection, we'll consider each individual employer's utility
as well as the sum of employers' utilities.  We think of this latter
sum as the {\em social welfare}, since it represents the total quality
of the applicants who are hired by any firm.
(For example, if all firms deterministically used the correct ranking,
then the top applicants would be the ones hired, leading to the highest
possible social welfare.)

\subsection{Modeling selection}

Each firm in our model has access to the same underlying pool of $n$ candidates,
which they rank using a randomized mechanism $\mc R$ to get a permutation $\pi$
as described above. Then, in a random order, each firm hires the highest-ranked
remaining candidate according to their ranking. Thus, if two firms both rank
candidate $i$ first, only one of them can hire $i$; the other must hire the next
available candidate according to their ranking. In our model, candidates
automatically accept the offer they get from a firm. For the sake of simplicity,
throughout this paper, we restrict ourselves to the case where there are two
firms hiring one candidate each, although our model readily generalizes to more
complex cases.

As described earlier, each firm can choose to use either a private ``human
evaluator'' or an algorithmically generated ranking as its randomized mechanism
$\mc R$. We assume that both candidate mechanisms come from a 
noisy permutation family $F_\theta$, with differing values of
the accuracy parameter $\theta$: human evaluators all have the same accuracy
$\theta_H$, and the algorithm has accuracy $\theta_A$. However, while the human
evaluator produces a ranking independent of any other firm, the algorithmically
generated ranking is identical for all firms who choose to use it. In other
words, if two firms choose to use the algorithmically generated ranking, they
will both receive the same permutation $\pi$.

The choice of which ranking mechanism to use leads to a game-theoretic setting:
both firms know the accuracy parameters of the human evaluators ($\theta_H$) and
the algorithm ($\theta_A$), and they must decide whether to use a human
evaluator or the algorithm. This choice introduces a subtlety: for many ranking
models, a firm's rational behavior depends not only on the accuracy of the
ranking mechanism, but also on the underlying candidate values $x_1, \dots,
x_n$. Thus, to fully specify a firm's behavior, we assume that $x_1, \dots, x_n$
are drawn from a known joint distribution $\mc D$. Our main results will hold
for any $\mc D$, meaning they apply even when the candidate values (but not
their identities) are deterministically known.

\subsection{Stating the main result}

Our main result is a pair of intuitive conditions under which a Braess'
Paradox-style result occurs---in other words, conditions under which 
there are accuracy parameters for which both firms
rationally choose to use the algorithmic ranking, but social welfare (and each
individual firm's utility) would be higher if both firms used independent human
evaluators.
Recall that the two firms hire in a random order. For a permutation $\pi$, let
$\pi_i$ denote the value of the $i$th-ranked candidate according to $\pi$.

We first state the two conditions, and then the theorem based on them.

\begin{definition}[Preference for the first position.]
A candidate distribution $\mc D$ and
noisy permutation family $\mc F_\theta$
exhibits a {\em preference for the first position} if 
for all
      $\theta > 0$, if $\pi, \sigma \sim \mc F_\theta$,
      \begin{equation*}
        \E{\pi_1 - \pi_2 \given \pi_1 \ne \sigma_1} > 0.
      \end{equation*}
\label{def:first}
\end{definition}
In other words, for any $\theta > 0$, suppose we draw two permutations $\pi$ and
$\sigma$ independently from $F_\theta$, and suppose that the first-ranked
candidates differ in $\pi$ and $\sigma$. Then the expected value of the
first-ranked candidate in $\pi$ is strictly greater than the expected value of
the second-ranked candidate in $\pi$.

\begin{definition}[Preference for weaker competition.]
A candidate distribution $\mc D$ and
noisy permutation family $\mc F_\theta$,
exhibits a {\em preference for weaker competition} if the following holds:
for all $\theta_1 > \theta_2$, $\sigma \sim \mc F_{\theta_1}$ and
      $\pi, \tau \sim \mc F_{\theta_2}$,
      \begin{equation*}
        \E{\pi_1^{(-\{\sigma_1\})}} < \E{\pi_1^{(-\{\tau_1\})}}.
      \end{equation*}
\label{def:weaker}
\end{definition}

Intuitively, suppose we have a higher accuracy parameter $\theta_1$ and a lower
accuracy parameter $\theta_2 < \theta_1$; we draw a permutation $\pi$ from
$\mc F_{\theta_2}$; and we then derive two permutations from $\pi$:
$\pi^{(-\{\sigma_1\})}$ obtained by deleting the first-ranked element of a
permutation $\sigma$ drawn from the more accurate distribution $\mc F_{\theta_1}$,
and $\pi^{(-\{\tau_1\})}$ obtained by deleting the first-ranked element of a
permutation $\tau$ drawn from the less accurate distribution $\mc F_{\theta_2}$.

Then the expected value of the first-ranked candidate in
$\pi^{(-\{\tau_1\})}$ is strictly greater than the 
expected value of the first-ranked candidate in
$\pi^{(-\{\sigma_1\})}$ --- that is,
when a random candidate is removed from $\pi$,
the best remaining candidate is better in
expectation when the randomly removed candidate is chosen based on a
noisier ranking.

Using these two conditions, we can state our theorem.

\begin{theorem}
  \label{thm:main}
  Suppose that a given candidate distribution $\mc D$ and 
  noisy permutation family $\mc F_\theta$
  satisfy Definition \ref{def:first} (preference for the first position) 
  and Definition \ref{def:weaker} (preference for weaker competition).

  Then, for any $\theta_H$, there exists $\theta_A > \theta_H$ such that
  using the algorithmic ranking is a strictly dominant strategy for both firms,
  but social welfare would be higher if both firms used human evaluators.
\end{theorem}

\subsection{A Preference for Independence}

Before we prove Theorem~\ref{thm:main}, we provide some intuition for the two
conditions in Definitions~\ref{def:first} and~\ref{def:weaker}. 
The second condition essentially says that it is better to have a
worse competitor: the firm randomly selected to hire second is better off if the
firm that hires first uses a less accurate ranking (in this case, a human
evaluator instead of the algorithmic ranking).

The first condition states that when two identically distributed permutations
disagree on their first element, the first-ranked candidate according to either
permutation is still better, in expectation, than the second-ranked candidate
according to either permutation. In what follows, we'll demonstrate that this
condition implies that firms in our model rationally prefer to make decisions
using independent (but equally accurate) rankings.

To do so, we need to introduce some notation. Recall that the two firms hire in
a random order. Given a candidate distribution $\mc D$, let $U_{s}(\theta_A,
\theta_H)$ denote the expected utility of the first firm to hire a candidate
when using ranking $s$, where $s \in \{A, H\}$ is either the algorithmic ranking
or the ranking generated by a human evaluator respectively. Similarly, let
$U_{s_1 s_2}(\theta_A, \theta_H)$ be the expected utility of the second firm to
hire given that the first firm used strategy $s_1$ and the second firm uses
strategy $s_2$, where again $s_1, s_2 \in \{A, H\}$. Finally, let $\pi, \sigma
\sim \mc F_{\theta}$.

In what follows, we will show that for any $\theta$,
\begin{equation}
  \E{\pi_1 - \pi_2 \given \pi_1 \ne \sigma_1} > 0
  \Longleftrightarrow U_{AH}(\theta, \theta) > U_{AA}(\theta, \theta).
  \label{eq:AH_AA_cond}
\end{equation}
In other words, whenever a ranking model meets Definition~\ref{def:first}, the
firm chosen to select second will prefer to use an \textit{independent} ranking
mechanism from it's competitor, given that the ranking mechanisms are equally
accurate.

First, we can write
\begin{align*}
  U_{AH}(\theta_A, \theta_H) &= \E{\pi_1 \cdot \ind{\pi_1 \ne \sigma_1} + \pi_2 \cdot \ind{\pi_1 =
  \sigma_1}} \\
  U_{AA}(\theta_A, \theta_H) &= \E{\sigma_2} \\
  &= \E{\sigma_2 \cdot \ind{\pi_1 \ne \sigma_1} + \sigma_2 \cdot \ind{\pi_1 =
  \sigma_1}}
\end{align*}
Thus,
\begin{align*}
  U_{AH}(\theta_A, \theta_H) &- U_{AA}(\theta_A, \theta_H) \\
  &= \E{(\pi_1 - \sigma_2) \cdot \ind{\pi_1 \ne \sigma_1} + (\pi_2 - \sigma_2)
  \cdot \ind{\pi_1 = \sigma_1}}.
\end{align*}
Conditioned on either $\pi_1 = \sigma_1$ or $\pi_1 \ne \sigma_1$, $\pi_2$ and
$\sigma_2$ are identically distributed and therefore have equal expectations.
As a result,
\begin{equation}
  \label{eq:UAH_UAA_diff}
  U_{AH}(\theta_A, \theta_H) - U_{AA}(\theta_A, \theta_H) = \E{(\pi_1 - \pi_2)
  \cdot \ind{\pi_1 \ne \sigma_1}},
\end{equation}
which implies~\eqref{eq:AH_AA_cond}. Thus, whenever a ranking model meets
Definition~\ref{def:first}, firms rationally prefer independent assessments, all
else equal.

To provide some intuition for what this preference for independence entails,
consider a setting where a hiring committee seeks to hire two candidates. They
meet, produce a ranking $\sigma$, and hire $\sigma_1$ (the best candidate
according to $\sigma$). Suppose they have the option to either hire $\sigma_2$ or
reconvene the next day to form an independent ranking $\pi$ and hire the best
remaining candidate according to $\pi$; which option should they choose? It's
not immediately clear why one option should be better than the other. However,
whenever Definition~\ref{def:first} is met, the committee should prefer to
reconvene and make their second hire according to a new ranking $\pi$. After
proving Theorem~\ref{thm:main}, we will provide natural ranking models that meet
Definition~\ref{def:first}, implying that under these ranking models,
independent re-ranking can be beneficial.

\subsection{Proving Theorem~\ref{thm:main}}

With this intuition, we are ready to prove Theorem~\ref{thm:main}.

\begin{proof}[Proof of Theorem~\ref{thm:main}]
  For given values of $\theta_A$ and $\theta_H$, using the algorithmic ranking
  is a strictly dominant strategy as long as
  \begin{align}
    U_A(\theta_A, \theta_H) + U_{AA}(\theta_A, \theta_H) &> U_H(\theta_A,
    \theta_H) + U_{AH}(\theta_A, \theta_H) \label{eq:dom1} \\
    U_A(\theta_A, \theta_H) + U_{HA}(\theta_A, \theta_H) &> U_H(\theta_A,
    \theta_H) + U_{HH}(\theta_A, \theta_H) \label{eq:dom2}
  \end{align}
  Note that~\eqref{eq:dom2} is always true for $\theta_A > \theta_H$ by the
  monotonicity assumption on $\mc F_\theta$: $U_A(\theta_A, \theta_H) \ge
  U_H(\theta_A, \theta_H)$ because a more accurate ranking produces a top-ranked
  candidate with higher expected value, and $U_{HA}(\theta_A, \theta_H) \ge
  U_{HH}(\theta_A, \theta_H)$ because this holds even conditioned on removing
  any candidate from the pool (in this case, the candidate randomly selected by
  the firm that hires first). Crucially, in~\eqref{eq:dom2}, the first firm's
  random selection is independent from the second firm's selection; the same
  logic could not be used to argue that~\eqref{eq:dom1} always holds for
  $\theta_A \ge \theta_H$. Moreover, when $\theta_A > \theta_H$, $U_A(\theta_A,
  \theta_H) > U_H(\theta_A, \theta_H)$ by the monotonicity assumption,
  meaning~\eqref{eq:dom2} holds.

  Let $W_{s_1 s_2}(\theta_A, \theta_H)$ denote social welfare when the two firms
  employ strategies $s_1, s_2 \in \{A, H\}$. Then, when both firms use the
  algorithmic ranking, social welfare is
  \begin{align*}
    W_{AA}(\theta_A, \theta_H) = U_A(\theta_A, \theta_H) + U_{AA}(\theta_A,
    \theta_H).
  \end{align*}

  By~\eqref{eq:AH_AA_cond}, Definition~\ref{def:first} implies that for any
  $\theta$, $U_{AA}(\theta, \theta) < U_{AH}(\theta, \theta)$, implying
  \begin{equation*}
    U_A(\theta_H, \theta_H) + U_{AA}(\theta_H, \theta_H) < U_H(\theta_H,
    \theta_H) + U_{AH}(\theta_H, \theta_H).
  \end{equation*}
  However, by the optimality assumption
  on $\mc F_\theta$ in Definition~\ref{def:family}, for sufficiently large $\hat
  \theta_A$,
  \begin{equation*}
    U_A(\hat \theta_A, \theta_H) + U_{AA}(\hat \theta_A, \theta_H) > U_H(\hat
    \theta_A, \theta_H) + U_{AH}(\hat \theta_A, \theta_H).
  \end{equation*}

  Note that $U_{s_1}(\theta_A, \theta_H)$ and $U_{s_1 s_2}(\theta_A,
  \theta_H)$ are continuous with respect to $\theta_A$ for any $s_1, s_2 \in \{A,
  H\}$ since they are expectations over discrete distributions with
  probabilities that are by assumption differentiable with respect to
  $\theta_A$. Therefore, by the Differentiability assumption on $\mc F_\theta$
  from Definition~\ref{def:family}, there is some $\theta_A^* > \theta_H$ such
  that
  \begin{equation}
    U_A(\theta_A^*, \theta_H) + U_{AA}(\theta_A^*, \theta_H) = U_H(\theta_A^*,
    \theta_H) + U_{AH}(\theta_A^*, \theta_H), \label{eq:indifferent}
  \end{equation}
  i.e., given that its competitor uses the algorithmic ranking, a firm is
  indifferent between the two strategies. For such $\theta_A^*$, using the
  algorithmic ranking is still a weakly dominant strategy. By definition of
  $W_{AA}$,
  \begin{align*}
    W_{AA}(\theta_A^*, \theta_H)
    &= U_H(\theta_A^*, \theta_H) + U_{AH}(\theta_A^*, \theta_H).
  \end{align*}
  If both firms had instead used human evaluators, social welfare would be
  \begin{align*}
    W_{HH}(\theta_A^*, \theta_H) = U_H(\theta_A^*, \theta_H) +
    U_{HH}(\theta_A^*, \theta_H).
  \end{align*}
  By Definition~\ref{def:weaker}, for $\sigma \sim \mc F_{\theta_{A^*}}$ and
  $\pi, \tau \sim \mc F_{\theta_H}$,
  \begin{equation*}
    \E{\pi_1^{(-\{\sigma_1\})}} < \E{\pi_1^{(-\{\tau_1\})}}.
  \end{equation*}
  Note that
  \begin{align*}
    U_{AH}(\theta_A^*, \theta_H) &= \E{\pi_1^{(-\{\sigma_1\})}} \\
    U_{HH}(\theta_A^*, \theta_H) &= \E{\pi_1^{(-\{\tau_1\})}}
  \end{align*}
  Thus, Definition~\ref{def:weaker} implies that for $\theta_{A^*} > \theta_H$,
  $U_{HH}(\theta_A^*, \theta_H) > U_{AH}(\theta_A^*, \theta_H)$. As a result for
  $\theta_{A^*} > \theta_H$, using the algorithmic ranking is a weakly dominant
  strategy, but
  \begin{align*}
    W_{HH}(\theta_A^*, \theta_H) 
    &= U_H(\theta_A^*, \theta_H) + U_{HH}(\theta_A^*, \theta_H) \\
    &> U_H(\theta_A^*, \theta_H) + U_{AH}(\theta_A^*, \theta_H) \\
    &= U_A(\theta_A^*, \theta_H) + U_{AA}(\theta_A^*, \theta_H) \\
    &= W_{AA}(\theta_A^*, \theta_H),
  \end{align*}
  meaning social welfare would have been higher had both firms used human
  evaluators.

  We can show that this effect persists for a value $\theta_A'$ such that using
  the algorithmic ranking is a \textit{strictly} dominant strategy. Intuitively,
  this is simply by slightly increasing $\theta_A^*$ so the algorithmic ranking
  is strictly dominant. For fixed $\theta_H$, define
  \begin{align*}
    f(\theta_A) &= U_A(\theta_A, \theta_H) + U_{AA}(\theta_A, \theta_H) \\
    g(\theta_A) &= U_H(\theta_A, \theta_H) + U_{AH}(\theta_A, \theta_H) \\
    h(\theta_A) &= U_H(\theta_A, \theta_H) + U_{HH}(\theta_A, \theta_H)
  \end{align*}
  Because~\eqref{eq:dom2} always holds for $\theta_A > \theta_H$, it suffices to
  show that there exists $\theta_A'$ such that $g(\theta_A') < f(\theta_A') <
  h(\theta_A')$. This is because $g(\theta_A') < f(\theta_A')$ is equivalent
  to~\eqref{eq:dom1} and $f(\theta_A') < h(\theta_A')$ is equivalent to
  $W_{AA}(\theta_A', \theta_H) < W_{HH}(\theta_A', \theta_H)$.

  First, note that $h(\theta_A)$ is a constant, and by
  Definition~\ref{def:weaker}, $g(\theta_A) < h(\theta_A)$ for all $\theta_A >
  \theta_H$. By the optimality assumption of Definition~\ref{def:family}, there
  exists sufficiently large $\hat \theta_A$ such that $f(\hat \theta_A) > g(\hat
  \theta_A)$. Recall that by definition of $\theta_A^*$, $f(\theta_A^*) =
  g(\theta_A^*)$. Both $f$ and $g$ are continuous by the Differentiability
  assumption in Definition~\ref{def:family}. Thus, there must exist some
  $\theta_A' > \theta_A^*$ such that $g(\theta_A') < f(\theta_A') <
  h(\theta_A')$. This means that for $\theta_A'$, using the algorithmic ranking
  is a strictly dominant strategy, but social welfare would still be larger if
  both firms used human evaluators.
\end{proof}

\section{Instantiating with Ranking Models}

Thus far, we have described a general set of conditions under which algorithmic
monoculture can lead to a reduction in social welfare. Under which ranking
models do these conditions hold? In the remainder of this paper, we instantiate
the model with two well-studied ranking models: Random Utility Models
(RUMs)~\cite{thurstone1927law} and the Mallows Model~\cite{mallows1957non}.
While RUMs do not always satisfy Definitions~\ref{def:first}
and~\ref{def:weaker}, they do under some realistic parameterizations, regardless
of the candidate distribution $\mc D$. Under the Mallows Model,
Definitions~\ref{def:first} and~\ref{def:weaker} are always met, meaning that
for any candidate distribution $\mc D$ and human evaluator accuracy $\theta_H$,
there exists an accuracy parameter $\theta_A$ such that a common algorithmic
ranking with accuracy $\theta_A$ decreases social welfare.

\begin{figure*}[ht]
  \centering
  \includegraphics[width=0.8\linewidth]{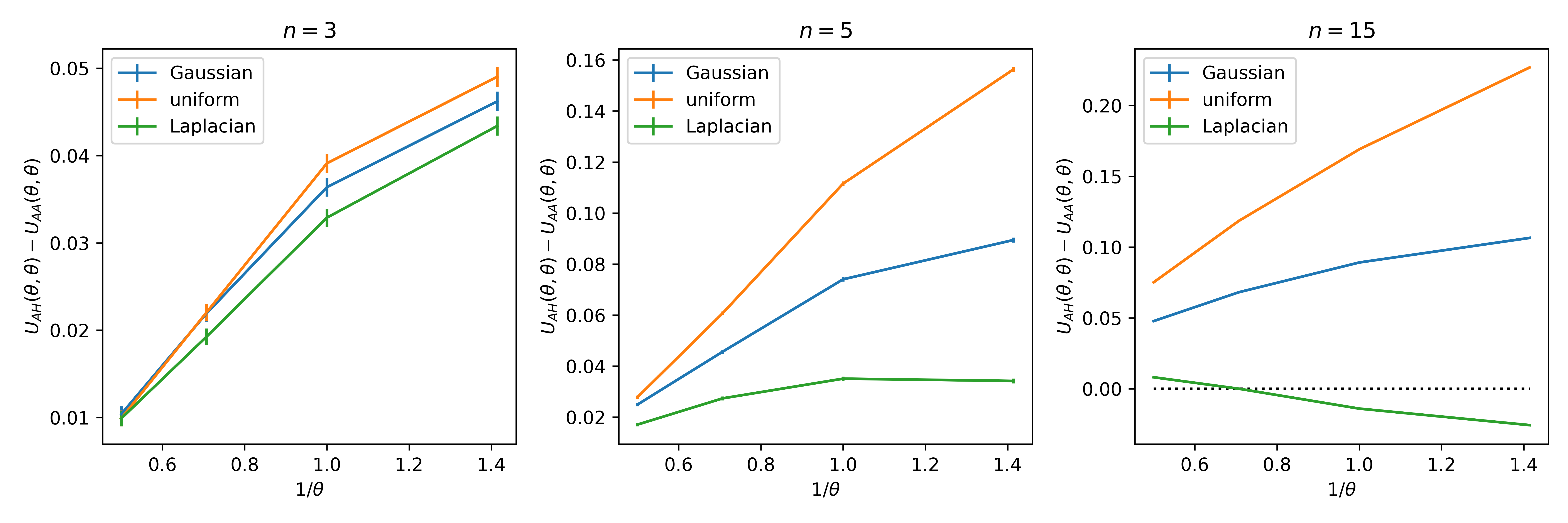}
  \caption{$U_{AH}(\theta, \theta) - U_{AA}(\theta, \theta)$ for three noise
    models with $n$ candidates whose utilities are drawn from a uniform
    distribution with unit variance for $n=3$, $n=5$, and $n=15$. Note that for
    $n=15$, $U_{AH}(\theta, \theta) - U_{AA}(\theta, \theta) < 0$ for Laplacian
  noise, meaning Definition~\ref{def:first} is not met.}
  \label{fig:UAH_UAA}
\end{figure*}

\begin{figure*}[ht]
  \centering
  \includegraphics[width=0.7\linewidth]{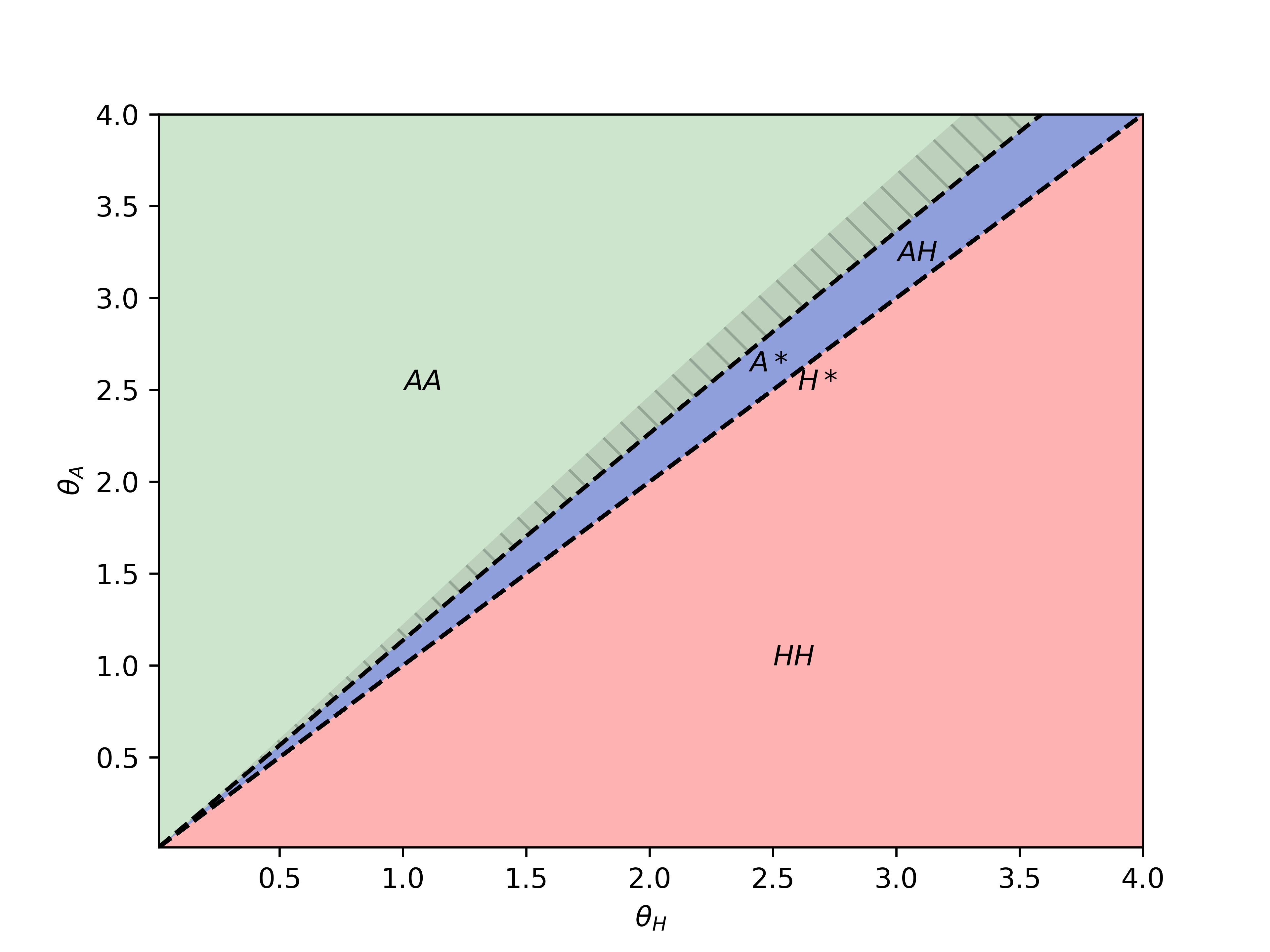}
  \caption{Regions for different equilibria. When human evaluators are more
    accurate than the algorithm, both firms decide to employ humans (HH). When
    the algorithm is significantly more accurate, both firms use the algorithm
    (AA). When the algorithm is slightly more accurate than human evaluators,
    two possible equilibria exist: (1) one firm uses the algorithm and the other
    employs a human (AH), or (2) both decide whether to use the algorithm with
    some probability $p$. The shaded portion of the green $AA$ region depicts
    where social welfare is smaller at the $AA$ equilibrium than it would be if
  both firms used human evaluators.}
  \label{fig:plane}
\end{figure*}

\subsection{Random Utility Models}

In Random Utility Models, the underlying candidate values $x_i$ are perturbed by
independent and identically distributed noise $\varepsilon_i \sim \mc E$, and
the perturbed values are ranked to produce $\pi$. Originally conceived in the
psychology literature~\cite{thurstone1927law}, this model has been well-studied
over nearly a
century,~\cite{daniels1950rank,block1960random,joe2000inequalities,yellott1977relationship,manski1977structure,strauss1979some},
including more recently in the computer science and machine learning
literature~\cite{azari2012random,soufiani2013generalized,ragain2016pairwise,zhao2018learning,makhijani2019parametric}.

First, we must define a family of RUMs that satisfies the conditions of
Definition~\ref{def:family}. Assume without loss of generality that the noise
distribution $\mc E$ has unit variance. Then, consider the family of RUMs
parameterized by $\theta$ in which candidates are ranked according to $x_i +
\frac{\varepsilon_i}{\theta}$. By this definition, the standard deviation of the
noise for a particular value of $\theta$ is simply $1/\theta$. Intuitively,
larger values of $\theta$ reduce the effect of the noise, making the ranking
more accurate.
In Theorem~\ref{thm:RUM_family} in Appendix~\ref{app:RUM_family}, we show as
long as the noise distribution $\mc E$ has positive support on $(-\infty,
\infty)$, this definition of $\mc F_\theta$ meets the differentiability,
asymptotic optimality, and monotonicity conditions in
Definition~\ref{def:family}. For distributions with finite support, many of our
results can be generalized by relaxing strict inequalities in
Definition~\ref{def:family} and Theorem~\ref{thm:main} to weak inequalities.

Because RUMs are notoriously difficult to work with analytically, we restrict
ourselves to the case where $n = 3$, i.e., there are 3 candidates. Under this
restriction, we can show that for Gaussian and Laplacian noise distributions,
Definition \ref{def:first} and \ref{def:weaker} --- the two
conditions of Theorem~\ref{thm:main} --- are met, regardless of the candidate
distribution $\mc D$. We defer the proof to Appendix~\ref{app:RUM_proof}.

\begin{theorem}
  \label{thm:RUM}
  Let $\mc F_\theta$ be the family of RUMs with either Gaussian or Laplacian
  noise with standard deviation $1/\theta$. Then, for any candidate distribution
  $\mc D$ over 3 candidates, the conditions of Theorem~\ref{thm:main} are
  satisfied.
\end{theorem}

It might be tempting to generalize Theorem~\ref{thm:RUM} to other distributions
and more candidates; however, certain noise and candidate distributions violate
the conditions of Theorem~\ref{thm:main}. Even for 3-candidate RUMs, there exist
distributions for which each of the conditions is violated; we provide such
examples in Appendix~\ref{app:counter}.

Moreover, while Gaussian and Laplacian distributions provably meet
Definitions~\ref{def:first} and~\ref{def:weaker} with only 3 candidates, this doesn't
necessarily extend to larger candidate sets. Figure~\ref{fig:UAH_UAA} shows that
Definition~\ref{def:first} can be violated under a particular candidate
distribution $\mc D$ for Laplacian noise with 15 candidates. This challenges the
intuition that independence is preferable---under some conditions, it can
actually better in expectation for a firm to use the \textit{same} algorithmic
ranking as its competitor, even if an independent human evaluator is equally
accurate overall. Unlike Theorem~\ref{thm:RUM}, which applies for \textit{any}
candidate distribution $\mc D$, certain noise models may violate
Definition~\ref{def:first} only for particular $\mc D$. It is an open question
as to whether Theorem~\ref{thm:RUM} can be extended to larger numbers of
candidates under Gaussian noise.

Finally, there exist noise distributions that violate Definition~\ref{def:first}
for any candidate distribution $\mc D$. In particular, the RUM family defined by
the Gumbel distribution is well-known to be equivalent to the Plackett-Luce
model of ranking, which is generated by sequentially selecting candidate $i$
with probability
\begin{equation}
  \label{eq:PL}
  \frac{\exp(\theta x_i)}{\sum_{j \in S} \exp(\theta x_j)},
\end{equation}
where $S$ is the set of remaining
candidates~\cite{luce1959individual,block1960random}. Under the Plackett-Luce
model, for any $\theta$, $U_{AH}(\theta, \theta) = U_{AA}(\theta, \theta)$. To
see this, suppose the firm that hires first selects candidate $i^*$. Then, the
firm that hires second gets each candidate $i$ with probability given
by~\eqref{eq:PL} with $S = \{1, \dots, n\} \backslash i^*$. As a result,
by~\eqref{eq:UAH_UAA_diff}, if $\pi, \sigma \sim \mc F_\theta$,
\begin{align*}
  \E{\pi_1 - \pi_2 \given \pi_1 \ne \sigma_1} = 0
\end{align*}
for any candidate distribution $\mc D$, meaning the Plackett-Luce model never
meets Definition~\ref{def:first}. Thus, under the Plackett-Luce model,
monoculture has no effect---the optimal strategy is always to use the best
available ranking, regardless of competitors' strategies.

Given the analytic intractability of most RUMs, it might appear that testing the
conditions of Theorem~\ref{thm:main}, especially for a particular noise and
candidate distributions, may not be possible; however, they can be efficiently
tested via simulation: as long as the noise distribution $\mc E$ and the
candidate distribution $\mc D$ can be sampled from, it is possible to test
whether the conditions of Theorem~\ref{thm:main} are satisfied. Thus, even if
the conditions of Theorem~\ref{thm:main} are not met for \textit{every}
candidate distribution $\mc D$, it is possible to efficiently determine whether
they are met for any \textit{particular} $\mc D$.

It is also interesting to ask about the magnitude of the negative impact
produced by monoculture. Our model allows for the qualities of candidates to be
either positive or negative (capturing the fact that a worker's productivity can
be either more or less than their cost to the firm in wages); using this, we can
construct instances of the model in which the optimal social welfare is positive
but the welfare under the (unique) monocultural equilibrium implied by Theorem 1
is negative. This is a strong type of negative result, in which sub-optimality
reverses the sign of the objective function, and it means that in general we
cannot compare the optimum and equilibrium by taking a ratio of two non-negative
quantities, as is standard in {\em Price of Anarchy} results. However, as a
future direction, it would be interesting to explore such Price of Anarchy
bounds in special cases of the problem where structural assumptions on the input
are sufficient to guarantee that the welfare at both the social optimum and the
equilibrium are non-negative. As one simple example, if the qualities for three
candidates are drawn independently from a uniform distribution centered at 0,
and the noise distribution is Gaussian, then there exist parameters $\theta_A >
\theta_H$ such that expected social welfare at the equilibrium where both firm
use the algorithmic ranking is non-negative, and approximately $4\%$ less than
it would be had both firms used human evaluators instead.

\subsection{The Mallows Model}

The Mallows Model also appears frequently in the ranking
literature~\cite{das2014role,lu2011learning}, and is much more analytically
tractable than RUMs. Under the Mallows Model, the likelihood of a permutation is
related to its distance from the true ranking $\pi^*$:
\begin{equation}
  \label{eq:mallows}
  \Pr[\pi] = \frac{1}{Z} \phi^{-d(\pi, \pi^*)},
\end{equation}
where $Z$ is a normalizing constant. In this model, $\phi > 1$ is the
accuracy parameter: the larger $\phi$ is, the more likely the ranking
procedure is to output a ranking $\pi$ that is close to the true ranking $r$. To
instantiate this model, we need a notion of distance $d(\cdot, \cdot)$ over
permutations. For this, we'll use Kendall tau distance, another standard notion
in the literature, which is simply the number of pairs of elements in $\pi$ that
are incorrectly ordered~\cite{kendall1938new}. In
Appendix~\ref{app:mallows_family}, we verify that the family of distributions
$\mc F_\theta$ given by the Mallows Model satisfies Definition~\ref{def:family},
defining $\theta = \phi - 1$ (for consistency, so $\theta$ is well-defined on
$(0, \infty)$).

In contrast to RUMs, the Mallows Model always satisfies the conditions of
Theorem~\ref{thm:main} for any candidate distribution $\mc D$, which we prove in
Appendix~\ref{app:mallows_proof}.

\begin{theorem}
  \label{thm:mallows}
  Let $\mc F_\theta$ be the family of Mallows Model distributions with parameter
  $\theta = \phi - 1$. Then, for any candidate distribution $\mc D$, the
  conditions of Theorem~\ref{thm:main} are satisfied.
\end{theorem}

Figure~\ref{fig:plane} characterizes firms' rational behavior at equilibrium in
the $(\theta_H, \theta_A)$ plane under the Mallows Model.
The decrease in social welfare found in
Theorem~\ref{thm:mallows} is depicted by the shaded portion of the green region
labeled $AA$, where social welfare would be higher if both firms used human
evaluators.

While the result of Theorem \ref{thm:mallows}
is certainly stronger than that of Theorem~\ref{thm:RUM}, 
in that it applies to all instances of the Mallows Model without
restrictions, it
should be interpreted with some caution. The Mallows Model does not depend on
the underlying candidate values, so according to this model, monoculture can
produce arbitrarily large negative effects. While insensitivity to candidate
values may not necessarily be reasonable in practice, our results hold for any
candidate distribution $\mc D$. Thus, to the extent that the Mallows Model can
reasonably approximate ranking in particular contexts, our results imply that
monoculture can have negative welfare effects.

\section{Models with Multiple Firms}

Our main focus in this work has been on models with two competing firms.
However, it is also interesting to consider the case of more than two firms;
we will see that the complex and sometimes counterintuitive effects
that we found in the two-firm case are further enriched by additional
phenomena.
Primarily, we will present the result of computational experiments with
the model, exposing some fundamental structural properties 
in the multi-firm problem
for which a formal analysis remains an intriguing open problem.
For concreteness, we will focus on a model in which
rankings are drawn
from the Mallows model.
As before, each firm must choose to order candidates according
to either an independent, human-produced ranking or an algorithmic
ranking common to all firms who choose it.
These rankings come from instances of the Mallows model with 
accuracy parameters $\phi_H$ and $\phi_A$ respectively as defined
in~\eqref{eq:mallows}.

\paragraph*{Braess' Paradox for $k > 2$ firms.}

First, we ask whether the Braess' Paradox effect persists with $k > 2$
firms. 
We find that it is possible to construct instances of
the problem with $k > 2$ for which Braess' Paradox occurs --- using
an algorithmic evaluation is a dominant strategy, but social welfare
would be higher if all firms used human evaluators instead. For
example, suppose $n = 4$, $k = 3$, $\phi_A = 2$, $\phi_H = 1.75$, and
candidate qualities are drawn from a uniform distribution on $[0, 1]$.
We find via computation that at equilibrium, 
each firm will rationally decide to use the
algorithmic evaluator and experience utility $\approx .551$, but if
all firms instead used human evaluators, they would experience utility
$\approx
.552$. Thus, Braess' Paradox can occur for $k > 2$ firms. Proving a
generalization of Theorem~\ref{thm:mallows}, to show that Braess'
Paradox can occur for any candidate distribution $\mc D$ and any value
of $\phi_H$ for $k > 2$ firms remains an open question. 

\paragraph*{Sequential decision-making.}

\begin{figure*}[ht]
  \centering
  \includegraphics[width=0.7\linewidth]{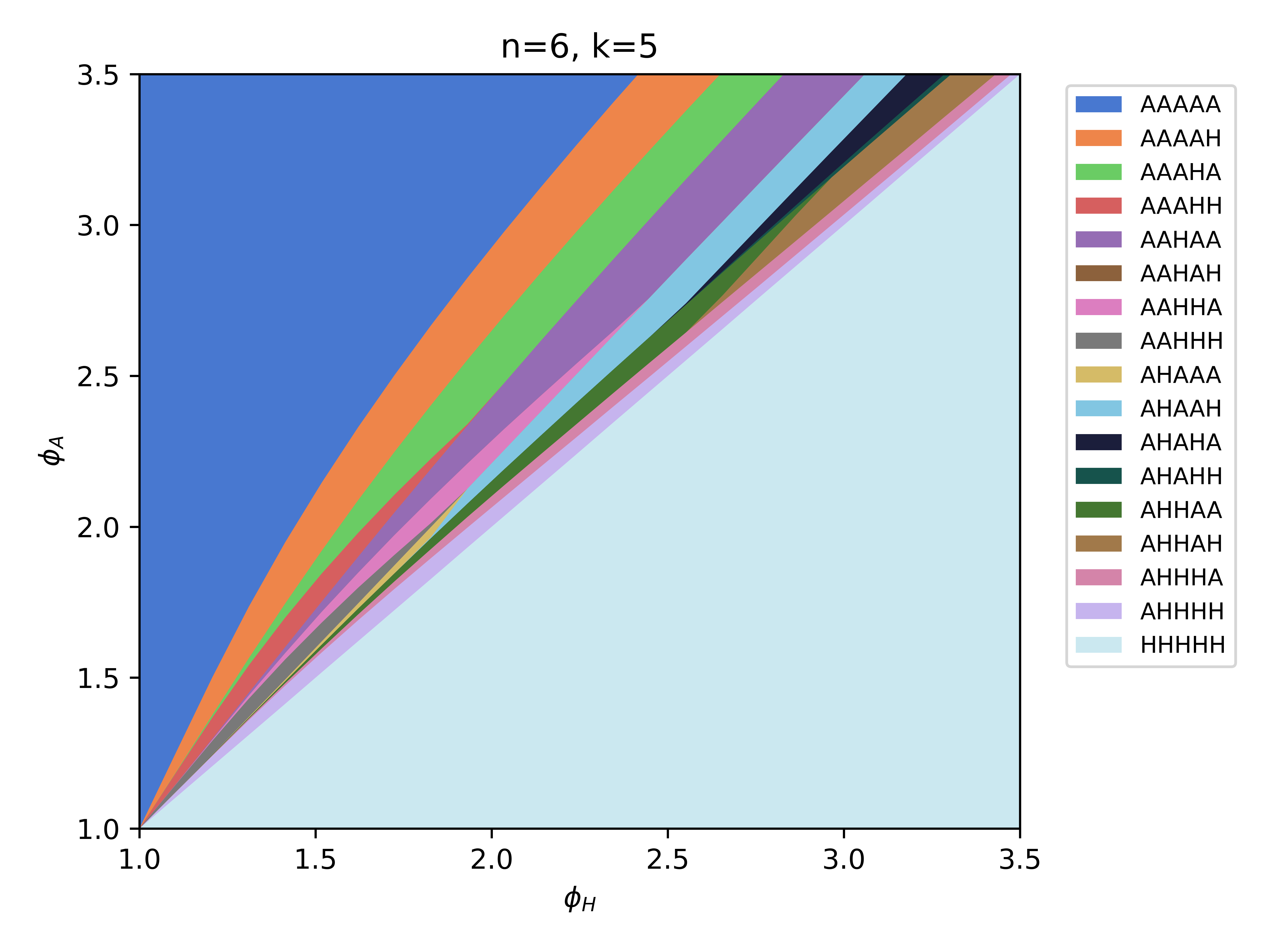}
  \caption{Regions for different optimal strategy profiles, where each strategy
    profile is a sequence of `$A$' and `$H$' representing the optimal strategies of
    each firm sequentially. For this plot, there are 5 firms ($k=5$) and 6
    candidates  ($n=6$) whose values are drawn from a uniform distribution. Note
    that when $\phi_A$ is much larger than $\phi_H$, all firms use the
    algorithmic ranking, but when $\phi_A$ is only slightly larger than
  $\phi_H$, only the first firm uses the algorithmic ranking.}%
  \label{fig:plane_k_firms}
\end{figure*}

Since the equilibrium behaviors we are studying take place in a model
where firms make decisions in a random order, a crucial first step is to
characterize firms' optimal behavior when making decisions
{\em sequentially}---that is, 
when firms hire in a fixed, known order as opposed to a random order.
In this context, consider the rational behavior of each firm: given a
distribution over candidate values, which ranking should each firm use? Clearly,
the first firm to make a selection should use the more accurate ranking
mechanism; however, as shown previously, subsequent firms' decisions are less
clear-cut. For a fixed number of firms, number of candidates, and distribution
over candidate values, we can explore the firms' optimal strategies 
over the possible space of $(\phi_H, \phi_A)$ values.

An optimal choice of strategies for the $k$ firms moving sequentially
can be written as a sequence of length $k$ made up of the symbols $A$ and $H$;
the $i^{\rm th}$ term in the sequence is equal to $A$ if the $i^{\rm th}$
firm to move sequentially uses the algorithm as its optimal strategy
(given the choices of the previous $i-1$ firms), and it is equal to $H$
if the $i^{\rm th}$ firm uses an independent human evaluation.
We can therefore represent the choice of optimal strategies, as 
the parameters $(\phi_H, \phi_A)$ vary, by a labeling of the 
$(\phi_H, \phi_A)$-plane: we label each point $(\phi_H, \phi_A)$
with the length-$k$ sequence that specifies the optimal sequence of strategies.

We can make the following initial formal observation about these optimal
sequences:
\begin{theorem}
  \label{thm:H_opt}
  When $\phi_H \ge \phi_A$, one optimal sequence is for all firms to choose $H$.
  When $\phi_H > \phi_A$, the unique optimal sequence is for all firms to choose
  $H$.
\end{theorem}
We prove this formally in Appendix~\ref{app:H_opt_proof}, but we provide a
sketch here.
When $\phi_H \geq \phi_A$, the first firm to move in sequence will simply
use the more accurate strategy, and hence will choose $H$.
Now, proceeding by induction, suppose that the first $i$ firms
have all chosen $H$, and consider the $(i+1)^{\rm st}$ firm to
move in sequence.
Regardless of whether this firm chooses $A$ or $H$, it will be
making a selection that is independent of the previous $i$ selections,
and hence it is optimal for it to choose $H$ as well.
Hence, by induction, it is an optimal solution for all firms to choose
$H$ when $\phi_H \geq \phi_A$.  
(This argument, slightly adapted, also directly establishes that 
it is uniquely optimal for all firms to choose $H$ when $\phi_H > \phi_A$.)

Beyond this observation, if we wish to extend
to the case when $\phi_A > \phi_H$,
the mathematical analysis of this multi-firm model remains an open question;
but it is possible to determine optimal strategies computationally for
each choice of $(\phi_H, \phi_A)$, and then to look at how these strategies 
vary over the $(\phi_H, \phi_A)$-plane.
Figure~\ref{fig:plane_k_firms} shows the result of doing this --- producing
a labeling of the $(\phi_H, \phi_A)$-plane as described above --- for 
$k = 5$ firms and $n = 6$ candidates, with the values of the candidates
drawn from a uniform distribution.

We observe a number of interesting phenomena from this labeling of the plane.
First, the region where $\phi_H \geq \phi_A$ is labeled with the
all-$H$ sequence, reflecting the argument above;
for the half-plane $\phi_A > \phi_H$, on the other hand, all optimal sequences
begin with $A$, since it is always optimal for the first firm 
to use the more accurate method.
The labeling of the half-plane $\phi_A > \phi_H$ becomes quite complex;
in principle, any sequence over the binary alphabet $\{A,H\}$
that begins with $A$ could be possible, and in fact we see that all $2^4 = 16$
of these sequences appear as labels in some portion of the plane.
This means that the sequential choice of optimal strategies for the
firms can display arbitrary non-monotonicities in the choice of 
algorithmic or human decisions, with firms alternating between them;
for example, even after the first firm chooses $A$ and the second chooses $H$,
the third may choose $A$ or $H$ depending on the values $(\phi_H,\phi_A)$.

The boundaries of the regions labeled by different optimal sequences
are similarly complex; some of the regions (such as $AAAHH$) appear
to be bounded, while others (such as $AHAHA$ and $AHHAH$) appear to 
only emerge for sufficiently large values of $\phi_H$.

Perhaps the most intriguing observation about the arrangement
of regions is the following.
Suppose we think of the sequences of symbols over $\{A,H\}$ as 
binary representations of numbers, with $A$ corresponding to the
binary digit $1$ and $H$ correspnding to the binary digit $0$.
(Thus, for example, $AAAHH$ would correspond to the number 
$16 + 8 + 4 = 28$, while $AHAHA$ would correspond to the number
$16 + 4 + 1 = 21$.)
The observation is then the following:
if we choose any vertical line $\phi_H = x$
(for a fixed $x$), and we follow it upward in the plane,
we encounter regions in increasing order of the numbers corresponding
to their labels, in this binary representation.
(First $HHHHH$, then $AHHHH$, then $AHHHA$, then $AHHAH$, and so forth.)

We do not know a proof for this fact, or how generally it holds,
but we can verify it computationally
for the regions of the $(\phi_H,\phi_A)$-plane mapped out in 
Figure~\ref{fig:plane_k_firms}, as well as similar computational experiments
not shown here for other choices of $k$ and $n$.
This binary-counter property suggests a rich body of additional structure
to the optimal strategies in the $k$-firm case, and we leave it as
an open question to analyze this structure mathematically.

\section{Conclusion}

Concerns about monoculture in the use of algorithms have focused on the danger
of unexpected, correlated shocks, and on the harm to particular individuals who
may fare poorly under the algorithm's decision. Our work here shows that
concerns about algorithmic monoculture are in a sense more fundamental, in that
it is possible for monoculture to cause decisions of globally lower average
quality, even in the absence of shocks. In addition to telling us something
about the pervasiveness of the phenomenon, it also suggests that it might be
difficult to notice its negative effects even while they're occurring --- these
effects can persist at low levels even without a shock-like disruption to call
our attention to them. Our results also make clear that algorithmic monoculture
in decision-making doesn't always lead to adverse outcomes; rather, we given
natural conditions under which such outcomes become possible, and show that
these conditions hold in a wide range of standard models.

Our results suggest a number of natural directions for further work. To begin
with, we have noted earlier in the paper that it would be interesting to give
more comprehensive quantitative bounds on the magnitude of monoculture's
possible negative effects in decisions such as hiring --- how much worse can the
quality of candidates be when selected with an equilibrium strategy involving
shared algorithms than with a socially optimal one? In formulating such
questions, it will be important to take into account how the noise model for
rankings relates to the numerical qualities of the candidates.

We have also focused here on the case of two firms and a single shared algorithm
that is available to both. It would be natural to consider generalizations
involving more firms and potentially more algorithms as well. With more
algorithms, we might see solutions in which firms cluster around different
algorithms of varying accuracies, as they balance the level of accuracy and the
amount of correlation in their decisions. It would also be interesting to
explore the ways in which correlations in firms' decisions can be decomposed
into constituent parts, such as the use of standardized tests that form input
features for algorithms, and how quantifying these forms of correlation might
help firms assess their decisions.

Finally, it will be interesting to consider how these types of results apply to
further domains. While the analysis presented here illustrates the consequences
of monoculture as applied to algorithmic hiring, our findings have potential
implications in a broader range of settings. Algorithmic monoculture not only
leads to a lack of heterogeneity in decision-making; by allowing valuable
options to slip through the cracks --- be they job candidates, potential hit
songs, or budding entrepreneurs --- it reduces total social welfare, even when
the individual decisions are more accurate on a case-by-case basis. These
concerns extend beyond the use of algorithms; whenever decision-makers rely on
identical or highly correlated evaluations, they miss out on hidden gems, and in
this way diminish the overall quality of their decisions.

\paragraph*{Acknowledgements.}
This work has been supported in part by a Simons Investigator Award,
a Vannevar Bush Faculty Fellowship,
a MURI grant,
an NSF Graduate Research Fellowship,
and 
grants from the ARO, AFOSR, and the MacArthur Foundation.

\bibliographystyle{plain}
\bibliography{refs}

\numberwithin{equation}{section}
\appendix

\section{Random Utility Models satisfying Definition~\ref{def:family}}
\label{app:RUM_family}
\begin{theorem}
  \label{thm:RUM_family}
  Let $f$ be the pdf of $\mc E$. The family of RUMs $\mc F_\theta$ given by
  ranking $x_i + \frac{\varepsilon_i}{\theta}$ with $\varepsilon_i \sim \mc E$
  satisfies the conditions of Definition~\ref{def:family} if:
  \begin{itemize}
    \item $f$ is differentiable
    \item $f$ has positive support on $(-\infty, \infty)$
  \end{itemize}
\end{theorem}
\begin{proof}
  We need to show that $\mc F_\theta$ satisfies the differentiability, asymptotic
  optimality, and monotonicity conditions in Definition~\ref{def:family}.

  \textbf{Differentiability:} The probability density of any realization of the $n$
  noise samples $\varepsilon_i/\theta$ is $\prod_{i=1}^n
  f(\varepsilon_i/\theta)$. Let $\varepsilon = [\varepsilon_1/\theta, \dots,
  \varepsilon_n/\theta]$ be the vector of noise values and let $M(\pi) \subseteq
  \mathbb{R}^n$ be the region such that any $\varepsilon \in M(\pi)$ will
  produces the ranking $\pi$. The probability of any permutation $\pi$ is
  \begin{align*}
    \Pr_\theta[\pi] = \int_{M(\pi)} \prod_{i=1}^n
    f\p{\frac{\varepsilon_i}{\theta}} \; d^n \mathbf{z}.
  \end{align*}
  Because $f$ is differentiable,
  \begin{align*}
    \frac{d}{d\theta} f\p{\frac{x}{\theta}}
    = f'\p{\frac{x}{\theta}} \cdot \p{-\frac{x}{\theta^2}}
  \end{align*}
  Because $\Pr{\theta}(\pi)$ is an integral of the product of differentiable
  functions over a fixed region, it is differentiable.

  \textbf{Asymptotic optimality:} We will show that for any pair of elements and
  any $\delta > 0$, there exists sufficiently large $\theta$ such that
  the probability that they incorrectly ranked is at most $\delta$. We will
  conclude with a union bound over the $n-1$ pairs of adjacent candidates that
  there exists sufficiently large $\theta$ such that the probability of
  outputting the correct ranking must be at least $1-(n-1)\delta$.

  Consider two candidates $x_i > x_{i+1}$. Let $\nu$ be the difference $x_i -
  x_{i+1}$. Then, they will be correctly ranked if
  \begin{align*}
    \frac{\varepsilon_i}{\theta} &> -\frac{\nu}{2} \\
    \frac{\varepsilon_{i+1}}{\theta} &< \frac{\nu}{2}
  \end{align*}
  Let $\overline{q}$ and $\underline{q}$ be the $1-\frac{\delta}{2}$ and
  $\frac{\delta}{2}$ quantiles of $\mc E$ respectively, and let $q =
  \max(|\overline{q}|, |\underline{q}|)$. For $\theta > \frac{2q}{\nu}$,
  \begin{align*}
    \Pr\b{\frac{\varepsilon_i}{\theta} < -\frac{\nu}{2}}
    &= \Pr\b{\varepsilon_i < -\frac{\nu\theta}{2}} \\
    &< \Pr\b{\varepsilon_i < -q} \\
    &\le \Pr\b{\varepsilon_i < \underline{q}} \\
    &= \frac{\delta}{2} \\
    \Pr\b{\frac{\varepsilon_{i+1}}{\theta} > \frac{\nu}{2}}
    &= \Pr\b{\varepsilon_{i+1} > \frac{\nu\theta}{2}} \\
    &< \Pr\b{\varepsilon_{i+1} > q} \\
    &\le \Pr\b{\varepsilon_{i+1} > \overline{q}} \\
    &= \frac{\delta}{2}
  \end{align*}
  Thus, for sufficiently large $\theta$, the probability that $x_i$ and
  $x_{i+1}$ are incorrectly ordered is at most $\delta$.

  Repeating this analysis for all $n-1$ pairs of adjacent elements, taking the
  maximum of all the $\theta$'s, and taking a union bound yields that the
  probability of incorrectly ordering any pair of elements is at most
  $(n-1)\delta$, meaning the probability of outputting the correct ranking is at
  least $1-(n-1)\delta$. Since $\delta$ is arbitrary, this probability can be
  made arbitrarily close to 1, satisfying the asymptotic optimality condition.

  \textbf{Monotonicity:} The removal of any elements does not alter the
  distribution of the remaining elements, meaning that the distribution of
  $\pi^{(-S)}$ is equivalent to a RUM with $n - |S|$ elements. Thus, it suffices
  to show that for a RUM with positive support on $(-\infty, \infty)$, the
  probability of ranking the best candidate first strictly increases with $\theta$.

  Recall that by definition, the candidates are ranked according to $x_i +
  \frac{\varepsilon_i}{\theta}$. The probability that $x_1$ is ranked first is
  \begin{align*}
    \Pr\b{x_1 + \frac{\varepsilon_1}{\theta} > \max_{2 \le i \le n} x_i +
    \frac{\varepsilon_i}{\theta}}
    &= \Pr\b{\frac{\varepsilon_1}{\theta} > \max_{2 \le i \le n} x_i - x_1 +
    \frac{\varepsilon_i}{\theta}} \\
    &= \Pr\b{\varepsilon_1 > \max_{2 \le i \le n} \theta(x_i - x_1) +
    \varepsilon_i} \\
    &= \mathbb{E}_{\varepsilon_2, \dots, \varepsilon_n} \Pr\b{\varepsilon_1 >
      \max_{2 \le i \le n} \theta(x_i - x_1) + \varepsilon_i \given
    \varepsilon_2, \dots, \varepsilon_n}
    \numberthis \label{eq:rum_monotonicity}
  \end{align*}
  We want to show that~\eqref{eq:rum_monotonicity} is increasing in $\theta$.
  Intuitively, this is because as $\theta$ increases, the right hand side of the
  inequality inside the probability decreases. To prove this formally, it
  suffices to show that the subderivative of~\eqref{eq:rum_monotonicity} with
  respect to $\theta$ only includes strictly positive numbers. First, we have
  \begin{align*}
    \frac{\partial}{\partial \theta} &\mathbb{E}_{\varepsilon_2, \dots,
      \varepsilon_n} \Pr\b{\varepsilon_1 > \max_{2 \le i \le n} \theta(x_i - x_1)
    + \varepsilon_i \given \varepsilon_2, \dots, \varepsilon_n} \subset \mathbb
    R_{>0} \\
                                     &\Longleftrightarrow
    \frac{\partial}{\partial \theta} \Pr\b{\varepsilon_1 > \max_{2 \le i \le n}
      \theta(x_i - x_1) + \varepsilon_i \given \varepsilon_2, \dots,
    \varepsilon_n} \subset \mathbb R_{>0}
  \end{align*}
  Let $F$ and $f$ be the cumulative density function and probability density
  function of $\mc E$ respectively. Then,
  \begin{equation*}
    \Pr\b{\varepsilon_1 > \max_{2 \le i \le n} \theta(x_i - x_1) + \varepsilon_i
    \given \varepsilon_2, \dots, \varepsilon_n}
    = 1 - F\p{\max_{2 \le i \le n} \theta(x_i - x_1) + \varepsilon_i}
  \end{equation*}
  Note that $F(\cdot)$ is strictly increasing (since $f$ is assumed to have
  positive support on $(-\infty, \infty)$), so it suffices to show that
  \begin{equation*}
    \frac{\partial}{\partial \theta} \max_{2 \le i \le n} \theta(x_i - x_1) +
    \varepsilon_i \subset \mathbb R_{<0}
  \end{equation*}
  For any $i$,
  \begin{equation*}
    \frac{d}{d\theta} \theta(x_i - x_1) + \varepsilon_i =  x_i - x_1 < 0.
  \end{equation*}
  Thus, the subderivative of the max of such functions includes only strictly
  negative numbers, which completes the proof.
\end{proof}

\section{3-candidate RUM Counterexamples}
\label{app:counter}

\subsection{Violating Definition~\ref{def:first}}

Here, we provide a noise mode $\mc E$, accuracy parameter $\theta$, and
candidate distribution $\mc D$ such that $U_{AH} < U_{AA}$.

Choose the noise distribution $\mc E$ and accuracy parameter $\theta$ such that
\begin{align*}
  \frac{\varepsilon}{\theta} &= \begin{cases}
    1 & w.p. \frac{\delta}{2} \\
    0 & w.p. 1-\delta \\
    -1 & w.p. \frac{\delta}{2}
  \end{cases}
\end{align*}
Note that this distribution does not satisfy Definition~\ref{def:family} because
it is neither differentiable nor supported on $(-\infty, \infty)$; however, we
can provide a ``smooth'' approximation to this distribution by expressing it as
the sum of arbitrarily tightly concentrated Gaussians with the same results.

We choose the candidate distribution $\mc D$ such that $x_1-1 > x_2 > x_3 >
x_1-2$. For example,
\begin{align*}
  x_1 &= \frac{7}{4} \\
  x_2 &= \frac{1}{2} \\
  x_3 &= 0
\end{align*}
Under this condition, assuming $x_3 = 0$ without loss of generality,
\begin{align*}
  U_{AH}(\theta, \theta) - U_{AA}(\theta, \theta) = \frac{\delta^{2}}{32}
  \p{\delta^{3} x_1 - 4 \delta^{2} x_1 + 4 \delta x_1 + 2 \delta^{3} x_2 - 14
  \delta^{2} x_2 + 20 \delta x_2 - 8 x_2}
\end{align*}
Notice that the lowest-power $\delta$ term is $-\frac{\delta^2 x_2}{4}$.
Therefore, for sufficiently small $\delta$, this is negative. For example,
plugging in the values given above with $\delta=.1$, $U_{AH}(\theta, \theta) -
U_{AA}(\theta, \theta) \approx -0.00076$.

\subsection{Violating Definition~\ref{def:weaker}}

Next, we'll give a 3-candidate RUM for which $U_{AH} < U_{HH}$ does not hold in
general. Consider the following 3-candidate example.

\begin{align*}
  x_1 &= 3 \\
  x_2 &= 2 \\
  x_3 &= 0
\end{align*}

Choose $\mc E$ and $\theta$ such that
\begin{equation*}
  \frac{\varepsilon}{\theta} = \begin{cases}
    1 & \text{w.p.} ~ \frac{1-\delta}{2} \\
    -1 & \text{w.p.} ~ \frac{1-\delta}{2} \\
    10 & \text{w.p.} ~ \frac{\delta}{2} \\
    -10 & \text{w.p.} ~ \frac{\delta}{2}
  \end{cases}
\end{equation*}
Again, while this noise model doesn't satisfy Definition~\ref{def:family}, we
can approximate it arbitrarily closely with the sum of tightly concentrated
Gaussians. Let the $\theta_A = 1.1 \theta$ and $\theta_H = 0.9\theta$.

We will show that for these parameters, $U_{AH}(\theta_A, \theta_H) >
U_{HH}(\theta_A, \theta_H)$, i.e., it is somehow better to choose after a better
opponent than after a worse opponent. At a high level, the reasoning for this is
as follows:
\begin{enumerate}
  \item When choosing first, the only difference between the algorithm and the
    human evaluator is that the algorithm is more likely to choose $x_2$ than
    $x_3$. Both strategies have identical probabilities of selecting $x_1$.
  \item When choosing second, the human evaluator's utility is higher when $x_2$
    has already been chosen than when $x_3$ has already been chosen. This is
    because when $x_2$ is unavailable, the human evaluator is almost guaranteed
    to get $x_1$; when $x_3$ is unavailable, the human evaluator will choose
    $x_2$ with probability $\approx 1/4$.
\end{enumerate}

Let $\tau$ and $\pi$ be rankings generated by the algorithm and human
evaluator respectively. First, we will show that
\begin{align}
  \Pr[\tau_1 = x_1] &= \Pr[\pi_1 = x_1] \label{eq:x1_same} \\
  \Pr[\tau_1 = x_2] &> \Pr[\pi_1 = x_2] \label{eq:x2_higher}
\end{align}
To do so, consider the realizations of $\varepsilon_1, \varepsilon_2,
\varepsilon_3$ that result in different rankings under $\theta_A$ and
$\theta_H$. In fact, the only set of realizations that result in different
rankings are when $\varepsilon_2/\theta = -1$ and $\varepsilon_3/\theta = 1$.
Thus, the algorithm and human evaluator always rank $x_1$ in the same position,
conditioned on a realization, which proves~\eqref{eq:x1_same}; the only
difference is that the algorithm sometimes ranks $x_2$ above $x_3$ when the
human evaluator does not. Moreover, whenever $\varepsilon_1/\theta = -10$, $x_2$
is more strictly more likely to be ranked first under the algorithm than the
human evaluator, which proves~\eqref{eq:x2_higher}.

Next, we must show that when choosing second, the human evaluator is better off
when $x_2$ is unavailable than when $x_3$ is unavailable. This is clearly true
because for the human evaluator,
\begin{align*}
  \Pr\b{x_1 + \frac{\varepsilon_1 \theta_H}{\theta} > x_3 + \frac{\varepsilon_3
  \theta_H}{\theta}} &\approx 1 - O(\delta) \\
  \Pr\b{x_1 + \frac{\varepsilon_1 \theta_H}{\theta} > x_2 + \frac{\varepsilon_3
  \theta_H}{\theta}} &\approx \frac{3}{4}
\end{align*}
Thus, conditioned on $x_2$ being unavailable, the human evaluator gets utility
$\approx 3$, whereas when $x_3$ is unavailable, the human evaluator gets utility
$\approx 2.75$. Let $u_{-i}$ be the expected utility for the human evaluator
when $x_i$ is unavailable. Putting this together, we get
\begin{align*}
  U_{AH}&(\theta_A, \theta_H) - U_{HH}(\theta_A, \theta_H) \\
  &= \sum_{i=1}^3 (\Pr[\tau_1 = x_i] - \Pr[\pi_1 = x_i]) u_{-i} \\
  &= (\Pr[\tau_1 = x_1] - \Pr[\pi_1 = x_1]) u_{-1}
  + (\Pr[\tau_1 = x_2] - \Pr[\pi_1 = x_2]) u_{-2}
  + (\Pr[\tau_1 = x_3] - \Pr[\pi_1 = x_3]) u_{-3}
  \\
  &= (\Pr[\tau_1 = x_2] - \Pr[\pi_1 = x_2]) u_{-2}
  + (\Pr[\tau_1 = x_3] - \Pr[\pi_1 = x_3]) u_{-3}
  \tag{$\Pr[\tau_1 = x_1 = \Pr[\pi_1 = x_1]$} \\
  &= (\Pr[\tau_1 = x_2] - \Pr[\pi_1 = x_2]) (u_{-2} - u_{-3})
  \tag{$\sum_{i=1}^3 \Pr[\tau_1 = x_i] = \sum_{i=1}^3 \Pr[\pi_1 = x_i]$} \\
  &> 0
\end{align*}
The last step follows from~\eqref{eq:x2_higher} and because $u_{-2} > u_{-3}$.

\section{Proof of Theorem~\ref{thm:RUM}}
\label{app:RUM_proof}

\subsection{Verifying Definition~\ref{def:first}}

By~\eqref{eq:AH_AA_cond}, we can equivalently show that for any $\theta$,
$U_{AH}(\theta, \theta) > U_{AA}(\theta, \theta)$. Let $\tau$ and $\pi$ be the
algorithmic and human-generated rankings respectively. Note that they're
identically distributed because $\theta_A = \theta_H$. Define
\begin{equation*}
  Y \triangleq \begin{cases}
    \pi_1 & \pi_1 \ne \tau_1 \\
    \pi_2 & \text{otherwise}
  \end{cases}
\end{equation*}
Note that $U_{AH}(\theta, \theta) = \E{Y}$ and $U_{AA}(\theta, \theta) =
\E{\tau_2}$. We want to show that $U_{AH}(\theta, \theta) - U_{AA}(\theta,
\theta) = \E{x_Y - x_{\tau_2}} > 0$. It is sufficient to show that for any $k$,
$\E{Y - \tau_2 \given \tau_1 = x_k} > 0$. Let $X_i = x_i +
\varepsilon_i/\theta$. Note that for distinct $i, j, k$ and $x_i > x_j$,
\begin{align*}
  \E{Y - \tau_2 \given \tau_1 = x_k} > 0
  &\Longleftarrow
  \frac{\Pr[Y = x_i \given \tau_1 = x_k]}{\Pr[Y = x_j \given \tau_1 = x_k]} >
  \frac{\Pr[\tau_2 = x_i \given \tau_1 = x_k]}{\Pr[\tau_2 = x_j \given \tau_1 = x_k]} \\
  &\Longleftrightarrow
  \Pr[Y = x_i \given \tau_1 = x_k] > \Pr[\tau_2 = x_i \given \tau_1 = x_k]
  \tag{numerator and denominator sum to 1}\\
  &\Longleftrightarrow
  \Pr[X_i > X_j] > \Pr[X_i > X_j \given X_k > X_i \cap X_k > X_j] \\
  &\Longleftrightarrow
  \Pr[X_i > X_j] > \mathbb{E}_{X_k}[\Pr[X_i > X_j \given X_k = a, X_i < a, X_j <
  a]].
\end{align*}
Thus, it suffices to show that for any $a$,
\begin{equation}
  \Pr[X_i > X_j] > \Pr[X_i > X_j \given X_i < a, X_j < a].
\end{equation}
Since $\Pr[X_i > X_j] = \lim_{a \to \infty} \Pr[X_i > X_j \given X_i < a, X_j <
a]$, it suffices to show that for all $a$,
\begin{equation}
  \label{eq:3cand_deriv}
  \frac{d}{da} \Pr[X_i > X_j \given X_i < a, X_j < a] \ge 0,
\end{equation}
and that it is strictly positive for some $a$. In other words, the higher $a$
is, the more likely $i$ and $j$ are to be correctly ordered. In
Theorems~\ref{thm:laplace_deriv} and~\ref{thm:gaussian_deriv}, we show
that~\eqref{eq:3cand_deriv} holds for both Laplacian and Gaussian noise
respectively, which proves that RUMs based on both distributions satisfy
Definition~\ref{def:first}.

\subsection{Verifying Definition~\ref{def:weaker}}

Next, we show that for both Laplacian and Gaussian distributions,
$U_{AH}(\theta_A, \theta_H) < U_{HH}(\theta_A, \theta_H)$ for all $\theta_A >
\theta_H$. In fact, for 3-candidate RUM families, we will show that this is
always true for any \textit{well-ordered} distribution, defined as follows.

\begin{definition}
  A noise model with density $f(\cdot)$ is \textbf{well-ordered} if for any
  $a>b$ and $c>d$,
  \begin{equation*}
    f(a-c)f(b-d) > f(a-d)f(b-c).
  \end{equation*}
\end{definition}

In other words, for a well-ordered noise model, given two numbers, two
candidates are more likely to be correctly ordered than inverted conditioned on
realizing those two numbers in some order. Lemma~\ref{lem:well-ordered} shows
that both Gaussian and Laplacian distributions are well-ordered.

Thus, it suffices to show that for any 3-candidate RUM with a well-ordered noise
model, $U_{AH}(\theta_A, \theta_H) < U_{HH}(\theta_A, \theta_H)$ when $\theta_A
> \theta_H$.
\begin{theorem}
  \label{thm:3cand_hh_ah}
  For 3 candidates with unique values $x_1 > x_2 > x_3$ and well-ordered i.i.d.
  noise with support $(-\infty, \infty)$, if $\theta_A > \theta_H$, then
  $U_{AH}(\theta_A, \theta_H) < U_{HH}(\theta_A, \theta_H)$.
\end{theorem}
\begin{proof}
  Define $u_{-i}$ to be the expected utility of the maximum element of the
  human-generated ranking when $i$ is not available. Because we're in the
  3-candidate setting, we have
  \begin{align*}
    u_{-1} &= \lambda_1 x_2 + (1-\lambda_1) x_3 \\
    u_{-2} &= \lambda_2 x_1 + (1-\lambda_2) x_3 \\
    u_{-3} &= \lambda_3 x_1 + (1-\lambda_3) x_2
  \end{align*}
  where $1/2 < \lambda_i < 1$. This is because the noise has support everywhere,
  so it is impossible to correctly rank any two candidates with probability 1,
  and any two candidates are more likely than not to be correctly ordered:
  \begin{align*}
    \Pr\b{\frac{\varepsilon_i}{\theta} - \frac{\varepsilon_j}{\theta} > -\delta}
    = \Pr[\varepsilon_i - \varepsilon_j \ge 0] + \Pr\b{0 > \frac{\varepsilon_i -
    \varepsilon_j}{\theta} > -\delta}
    > \frac{1}{2}
  \end{align*}
  Note that $\lambda_2 > \lambda_1$ and $\lambda_2 >
  \lambda_3$, since
  \begin{align*}
    \lambda_2 = \Pr[\varepsilon_1 - \varepsilon_3 > -\theta(x_1 - x_3)]
    > \max\left\{\Pr[\varepsilon_1 - \varepsilon_3 > -\theta(x_2 - x_3)],
    \Pr[\varepsilon_1 - \varepsilon_3 > -\theta(x_1 - x_2)]\right\}
    = \max\{\lambda_1, \lambda_3\}.
  \end{align*}
  Let $\tau \sim \mc F_{\theta_A}$ and $\pi \sim \mc F_{\theta_H}$. With this,
  we can write
  \begin{align*}
    U_{AH}(\theta_A, \theta_H) &= \sum_{i=1}^3 \Pr[\tau_1 = i] u_{-i} \\
    U_{HH}(\theta_A, \theta_H) &= \sum_{i=1}^3 \Pr[\pi_1 = i] u_{-i}
  \end{align*}
  Define
  \begin{equation*}
    \Delta p_i = \Pr[\tau_1 = i] - \Pr[\pi_1 = i]
  \end{equation*}
  Using Lemmas~\ref{lem:AnvsHn}, and~\ref{lem:deltapi}, we have
  \begin{align*}
    \Delta p_1 &> 0 \tag{By monotonicity of RUM families, see
    Appendix~\ref{app:RUM_family}} \\
    \Delta p_1 &\ge \Delta p_2 \\
    \Delta p_3 &\le 0
  \end{align*}
  Also, $\Delta p_1 + \Delta p_2 + \Delta p_3 = 0$.
  We must show that
  \begin{equation*}
    U_{AH}(\theta_A, \theta_H) - U_{HH}(\theta_A, \theta_H) = \sum_{i=1}^3
    \Delta p_i u_{-i} < 0.
  \end{equation*}
  We consider 2 cases.

  \textbf{Case 1:} $\Delta p_2 \le 0$. \\
  Then, $\Delta p_1 = -(\Delta p_2 + \Delta p_3)$. This yields
  \begin{align*}
    \sum_{i=1}^3 \Delta p_i u_{-i}
    &= \Delta p_1 u_{-1} + \Delta p_2 u_{-2} + \Delta p_3 u_{-3} \\
    &\le \Delta p_1 u_{-1} - \Delta p_1 \min(u_{-2}, u_{-3}) \\
    &= \Delta p_1 \p{\lambda_1 x_2 + (1-\lambda_1) x_3 - \min\left\{\lambda_2
    x_1 + (1-\lambda_2) x_3, \lambda_3 x_1 + (1-\lambda_3) x_2\right\}} \\
    &\le \Delta p_1 \p{\lambda_1 x_2 + (1-\lambda_1) x_3 - \min\left\{\lambda_2
    x_1 + (1-\lambda_2) x_3, x_2\right\}}
  \end{align*}
  We can show that this is at most 0 regardless of which term attains the
  minimum. Because $\lambda_2 > \lambda_1$,
  \begin{align*}
    \lambda_1 x_2 + (1-\lambda_1) x_3 - \lambda_2 x_1 - (1-\lambda_2) x_3
    &= \lambda_1 x_2 + x_3 - \lambda_1 x_3 - \lambda_2 x_1 - x_3 + \lambda_2 x_3
    \\
    &= \lambda_1 x_2 - \lambda_1 x_3 - \lambda_2 x_1 + \lambda_2 x_3 \\
    &= \lambda_1 (x_2 - x_3) + \lambda_2 (x_3 - x_1) \\
    &< \lambda_1 (x_2 - x_3) + \lambda_1 (x_3 - x_1) \\
    &= \lambda_1(x_2 - x_1) \\
    & < 0
  \end{align*}
  For the second term, we have
  \begin{equation*}
    \lambda_1 x_2 + (1-\lambda_1) x_3 - x_2 = (1-\lambda_1)(x_3 - x_2) < 0.
  \end{equation*}
  Thus,
  \begin{align*}
    \sum_{i=1}^3 \Delta p_i u_{-i} < 0.
  \end{align*}

  \textbf{Case 2:} $\Delta p_2 > 0$.
  Note that $u_{-1} < x_2 < u_{-3}$. Then, using $\Delta p_3 = -(\Delta p_1
  + \Delta p_2)$,
  \begin{align*}
    \sum_{i=1}^3 \Delta p_i u_{-i}
    &= \Delta p_1 u_{-1} + \Delta p_2 u_{-2} + \Delta p_3 u_{-3} \\
    &= \Delta p_1 (u_{-1} - u_{-3}) + \Delta p_2 (u_{-2} - u_{-3}) \\
    &\le \Delta p_2 (u_{-1} - u_{-3}) + \Delta p_2 (u_{-2} - u_{-3})
    \tag{$\Delta p_1 \ge \Delta p_2$ and $u_{-1} < u_{-3}$} \\
    &= \Delta p_2 (u_{-1} + u_{-2} - 2u_{-3}) \\
    &\le \Delta p_2 (x_2 + x_1 - 2 (\lambda_3 x_1 + (1-\lambda_3) x_2)) \\
    &< \Delta p_2 \p{x_2 + x_1 - 2 \p{\frac{1}{2} x_1 + \frac{1}{2} x_2}}
    \tag{$\lambda_3 > \frac{1}{2}$} \\
    &= 0
  \end{align*}

  Thus, $U_{AH}(\theta_A, \theta_H) < U_{HH}(\theta_A, \theta_H)$.
\end{proof}

\subsection{Supplementary Lemmas for Random Utility Models}
\begin{lemma}
  \label{lem:well-ordered}
  Both Gaussian and Laplacian distributions are well-ordered.
\end{lemma}
\begin{proof}
  The Gaussian noise model is
  well-ordered:
  \begin{align*}
    f(a-c)f(b-d)
  &= \frac{1}{2 \sigma^2 \pi} \exp(-(a-c)^2 -(b-d)^2) \\
  &= \frac{1}{2 \sigma^2 \pi} \exp(-(a-d)^2 - (b-c)^2 - 2(ac + bd - ad - bc)) \\
  &= f(a-d)f(b-c) \exp(-2((a-b)(c-d))) \\
  &< f(a-d)f(b-c)
  \end{align*}
  Laplacian noise is as well:
  \begin{align*}
    f(a-c)f(b-d) &= \frac{1}{4} \exp(-|a-c|-|b-d|) \\
    f(a-d)f(b-c) &= \frac{1}{4} \exp(-|a-d|-|b-c|) \\
  \end{align*}
  It suffices to show that for $a>b$ and $c>d$, $|a-c|+|b-d| < |a-d|+|b-c|$.
  To show this, plot $(a, b)$ and $(c, d)$ in the $(x, y)$ plane. Note that
  they're both below the $y=x$ line, and that the $\ell_1$ distance between them
  is $|a-c|+|b-d|$. Moreover, the $\ell_1$ distance between any two points must be
  realized by some Manhattan path, which is a combination of horizontal and
  vertical line segments. Consider the point $(b, a)$, which is above the $y=x$
  line. Any Manhattan path from $(b, a)$ to $(c, d)$ must cross the $y=x$ line at
  some point $(w, w)$. Since $(b, a)$ and $(a, b)$ are equidistant from $(w, w)$,
  for any Manhattan path from $(b, a)$ to $(c, d)$, there exists a Manhattan path
  from $(a, b)$ to $(c, d)$ passing through $(w, w)$  of the same length, meaning
  the $\ell_1$ distance from $(a, b)$ to $(c, d)$ is smaller than the $\ell_1$
  distance from $(b, a)$ to $(c, d)$. As a result, $|a-c|+|b-d| < |a-d|+|b-c|$.
\end{proof}

Next, we show a few basic facts. Let $f_A(r)$ be the density function of the
joint realization $R = [X_1, \dots, X_n] = [x_1 + \varepsilon_1/\theta_A, \dots,
x_n + \varepsilon_n/\theta_A]$ under the algorithmic ranking and $f_H(r)$ be the
similarly defined density function under the human-generated ranking. Consider
the ``contraction'' operation $r' = \cont(r)$ such that $r_i'= x_i + (r_i-x_i)
\cdot \frac{\theta_H}{\theta_A}$. Essentially, the contraction defines a
coupling between $f_A(\cdot)$ and $f_H(\cdot)$, since for $r' = \cont(r)$,
$f_A(r') \dr' = f_H(r) \dr$. Let $\pi(r)$ be the ranking induced by $r$. Note
that contraction cannot introduce any new inversions in $\pi(r)$---that is, if
$i$ is ranked above $j$ in $\pi(r)$ for $i < j$, then $i$ is ranked above $j$ in
$\pi(\cont(r))$. Intuitively, this is because contraction pulls values closer to
their means, and can therefore only correct existing inversions, not introduce
new ones. This fact will allow us to prove some useful lemmas.

\begin{lemma}
  If $F_{\theta}$ is a RUM family satisfying Definition~\ref{def:family}, then
  for $\tau \sim \mc F_{\theta_A}$ and $\pi \sim \mc F_{\theta_H}$,
  \begin{equation*}
    \Pr[\tau_1 = x_n] \le \Pr[\pi_1 = x_n]
  \end{equation*}
  \label{lem:AnvsHn}
\end{lemma}
\begin{proof}
  Consider any realization $r$. Because inversions can only be corrected, not
  generated, by contraction, if $\pi_1(r') = n$, then $\pi_1(r) = n$ where $r' =
  \cont(r)$. Since $r'$ and $r$ have equal measure under $f_A$ and $f_H$
  respectively, we have
  \begin{align*}
    \Pr[\pi_1 = x_n]
    &= \int_{\R^n} f_H(r) \ind{\pi_1(r) = x_n} \dr \\
    &= \int_{\R^n} f_A(\cont(r)) \ind{\pi_1(r) = x_n} \; d\cont(r) \\
    &\ge \int_{\R^n} f_A(\cont(r)) \ind{\pi_1(\cont(r)) = x_n} \; d\cont(r) \\
    &= \int_{\R^n} f_A(r) \ind{\pi_1(r) = x_n} \dr \\
    &= \Pr[\tau_1 = x_n]
  \end{align*}
\end{proof}

Next, we prove the following result for well-ordered noise models.
\begin{lemma}
  For any $i > 1$, if the noise model $\mc E$ is well-ordered, for $\theta_A \ge
  \theta_H$, $\tau \sim \mc F_{\theta_A}$, and $\pi \sim \mc F_{\theta_H}$,
  \begin{equation*}
    \Pr[\tau_1 = x_1] - \Pr[\pi_1 = x_1] \ge \Pr[\tau_1 = x_i] - \Pr[\pi_1
    = x_i]
  \end{equation*}
  \label{lem:deltapi}
\end{lemma}
\begin{proof}
  For $j \ne i$, let $S_{j \to i} \subseteq \R^n$ be the set of realizations $r$
  such that $\pi_1(r) = x_j$ and $\pi_1(\cont(r)) = x_i$. Note that $S_{j \to i}
  = \emptyset$ for $j < i$ because contraction cannot create inversions. Then,
  we have that
  \begin{equation*}
    \Pr[\tau_1 = x_i] - \Pr[\pi_1 = x_i]
    = \sum_{j > i} \int_{\R^n} f_H(r) \ind{r \in S_{j \to i}} \dr - \sum_{j <
    i} \int_{R^n} f_H(r) \ind{r \in S_{i \to j}} \dr
    \le \sum_{j > i} \int_{\R^n} f_H(r) \ind{r \in S_{j \to i}} \dr
  \end{equation*}
  Define
  \begin{equation*}
    \swap_i(r) = r',
  \end{equation*}
  where
  \begin{equation*}
    r_j' = \begin{cases}
      r_j & j \notin \{1, i\} \\
      r_1 & j = i \\
      r_i & j = 1
    \end{cases}
  \end{equation*}
  Intuitively, the $\swap_i$ operation simply swaps the realizations in
  positions $1$ and $i$. Note that this is a bijection. Also, if $r \in S_{j \to
  i}$, then $\swap_i(r) \in S_{j \to 1}$, since
  \begin{align*}
    \cont(\swap_i(r))_1 &\ge \cont(r)_i \ge \max_j \cont(r)_j \ge \max_{j \notin
    \{1, i\}} \cont(\swap_i(r))_j \\
    \cont(\swap_i(r))_1 &\ge \cont(r)_i \ge  \cont(r)_1 \ge \cont(\swap_i(r))_i
  \end{align*}
  Furthermore for $r \in S_{j \to i}$, $f_H(r) \le f_H(\swap_i(r))$ since
  \begin{align*}
    \frac{f_H(\swap_i(r))}{f_H(r)} &= \frac{f(r_i - x_1) f(r_1 - x_i)}{f(r_1 -
    x_1)f(r_i - x_i)} \ge 1
  \end{align*}
  because the noise is well-ordered and $r \in S_{j \to i}$ implies $r_i > r_1$.
  Thus,
  \begin{align*}
    \sum_{j > i} \int_{\R^n} f_H(r) \ind{r \in S_{j \to i}} \dr
    &\le \sum_{j > i} \int_{\R^n} f_H(\swap_i(r)) \ind{r \in S_{j \to i}} \dr \\
    &\le \sum_{j > i} \int_{\R^n} f_H(\swap_i(r)) \ind{\swap_i(r) \in S_{j \to 1}} \dr \\
    &\le \sum_{j > i} \int_{\R^n} f_H(r) \ind{r \in S_{j \to 1}} \dr \\
    &\le \sum_{j > 1} \int_{\R^n} f_H(r) \ind{r \in S_{j \to 1}} \dr \\
    &= \Pr[\tau_1 = x_1] - \Pr[\pi_1 = x_1]
  \end{align*}
\end{proof}

Finally, we show that~\eqref{eq:3cand_deriv} holds for both Laplacian and
Gaussian noise.

\begin{theorem}
  For any $a \in \mathbb R$ and $X_i = x_i + \sigma \varepsilon_i$ where
  $\varepsilon_i$ is Laplacian with unit variance,
  \begin{equation*}
    \frac{d}{da} \Pr[X_i > X_j \given X_i < a, X_j < a] \ge 0.
  \end{equation*}
  Moreover, it is strictly positive for some $a$.
  \label{thm:laplace_deriv}
\end{theorem}

\begin{proof}
First, we must derive an expression for $\Pr[X_i > X_j \given X_i < a, X_j <
a]$. Recall that the Laplace distribution parameterized by $\mu$ and $\lambda$
has pdf
\begin{equation*}
  f(x; \mu, \lambda) = \frac{\lambda}{2} \exp(-\lambda|x-\mu|)
\end{equation*}
and cdf
\begin{equation*}
  F(x; \mu, \lambda) = \begin{cases}
    \frac{1}{2} \exp\p{-\lambda(\mu-x)} & x < \mu \\
    1 - \frac{1}{2} \exp\p{-\lambda(x-\mu)} & x \ge \mu
  \end{cases}
\end{equation*}
Note that $x_i$ and $x_j$ be the respective means of $X_i$ and $X_j$, with $x_i
> x_j$. Because the Laplace distribution is piecewise defined, we must consider
3 cases and show that in all 3 cases, \eqref{eq:3cand_deriv} holds. Note that
\begin{equation}
  \label{eq:lap_integral}
  \Pr[X_i > X_j \given X_i < a, X_j < a] = \frac{\int_{-\infty}^a
    f(x; x_i, \lambda) F(x; x_j, \lambda) \dx}{F(a; x_i, \lambda) F(a;
    x_j, \lambda)}
\end{equation}

\textbf{Case 1:} $a \le x_j$. \\
Then, the numerator of~\eqref{eq:lap_integral} is
\begin{align*}
  \int_{-\infty}^a \frac{\lambda}{2} \exp\p{-\lambda(x_i-x)} \cdot
  \frac{1}{2} \exp(-\lambda(x_j - x)) \dx
  &= \frac{\lambda}{4} \int_{-\infty}^a \exp(-\lambda(x_i + x_j - 2x)) \dx
  \\
  &= \frac{\lambda\exp(-\lambda(x_i + x_j))}{4} \int_{-\infty}^a
  \exp(2\lambda x) \dx \\
  &= \frac{\lambda\exp(-\lambda(x_i + x_j))}{4} \frac{1}{2\lambda}
  \exp(2\lambda a) \\
  &= \frac{\exp(-\lambda(x_i + x_j - 2a))}{8}
\end{align*}
The denominator is
\begin{align*}
  \frac{1}{2} \exp(-\lambda(x_i - a)) \cdot \frac{1}{2} \exp(-\lambda(x_j -
  a))
  &= \frac{1}{4} \exp(-\lambda(x_i + x_j - 2a)).
\end{align*}
Thus,
\begin{equation*}
  \Pr[X_i > X_j \given X_i < a, X_j < a] = \frac{1}{2},
\end{equation*}
so its derivative is trivially nonnegative.

\textbf{Case 2:} $x_j < a \le x_i$. \\
Then, the numerator of~\eqref{eq:lap_integral} is
\begin{align*}
  \int_{-\infty}^{x_j}& \frac{\lambda}{2} \exp\p{-\lambda(x_i-x)} \cdot
  \frac{1}{2} \exp(-\lambda(x_j - x)) \dx + 
  \int_{x_j}^{a} \frac{\lambda}{2} \exp\p{-\lambda(x_i-x)}
  \p{1 - \frac{1}{2} \exp(-\lambda(x-x_j))} \dx \\
  &= \frac{\exp(-\lambda(x_i - x_j))}{8} + \frac{\lambda}{2} \int_{x_j}^a
  \exp(-\lambda(x_i-x)) \dx - \frac{\lambda}{4}\int_{x_j}^a
  \exp(-\lambda(x_i - x_j)) \dx \\
  &= \frac{\exp(-\lambda(x_i - x_j))}{8} + \frac{\lambda}{2}
  \frac{1}{\lambda} (\exp(-\lambda(x_i-a)) - \exp(-\lambda(x_i - x_j))) -
  \frac{\lambda}{4} (a-x_j) \exp(-\lambda(x_i - x_j)) \\
  &= \frac{1}{2} \exp(-\lambda(x_i-a)) - \p{\frac{3}{8} + \frac{\lambda}{4}
  (a-x_j)} \exp(-\lambda(x_i - x_j))
\end{align*}
The denominator is
\begin{align*}
  \p{1 - \frac{1}{2} \exp(-\lambda(a-x_j))} \cdot \frac{1}{2}
  \exp(\lambda(x_j - a))
  = \frac{1}{2} \exp(-\lambda(x_i - a)) - \frac{1}{4} \exp(-\lambda(x_i -
  x_j))
\end{align*}
We can factor out $\frac{1}{4}\exp(-\lambda(x_i-x_j))$ from both, so
\begin{align*}
  \Pr[X_i > X_j \given X_i < a, X_j < a]
  &= \frac{2\exp(\lambda(a-x_j)) - \p{\frac{3}{2} + \lambda
  (a-x_j)}}{2\exp(\lambda(a-x_j)) - 1} \\
  &= \frac{2\exp(\lambda(a-x_j)) - 1 - \p{\frac{1}{2} + \lambda
  (a-x_j)}}{2\exp(\lambda(a-x_j)) - 1} \\
  &= 1 - \frac{\frac{1}{2} + \lambda (a-x_j)}{2\exp(\lambda(a-x_j)) - 1}
\end{align*}
Thus,
\begin{align*}
  &\frac{d}{da} \Pr[X_i > X_j \given X_i < a, X_j < a] > 0 \\
  \Longleftrightarrow &
  \frac{d}{da} \frac{\frac{1}{2} + \lambda (a-x_j)}{2\exp(\lambda(a-x_j)) -
  1} < 0 \\
  \Longleftrightarrow &
  (2\exp(\lambda(a-x_j))-1)\lambda < \p{\frac{1}{2} + \lambda(a-x_j)}
  2\lambda\exp(\lambda(a-x_j)) \\
  \Longleftrightarrow &
  2 - \exp(-\lambda(a-x_j)) < 2\p{\frac{1}{2} + \lambda(a-x_j)} \\
  \Longleftrightarrow &
  1 - \exp(-\lambda(a-x_j)) < 2\lambda(a-x_j) \\
  \Longleftrightarrow &
  \exp(-\lambda(a-x_j)) > 1 - 2\lambda(a-x_j)
\end{align*}
This is true because $\lambda(a-x_j) > 0$, and for $z > 0$,
\begin{equation*}
  \exp(-z) > 1-z > 1-2z.
\end{equation*}

\textbf{Case 3:} $a > x_i$. \\
Then, the numerator of~\eqref{eq:lap_integral} is
\begin{align*}
  \int_{-\infty}^{x_j}& \frac{\lambda}{2} \exp\p{-\lambda(x_i-x)} \cdot
  \frac{1}{2} \exp(-\lambda(x_j - x)) \dx + 
  \int_{x_j}^{x_i} \frac{\lambda}{2} \exp\p{-\lambda(x_i-x)}
  \p{1 - \frac{1}{2} \exp(-\lambda(x-x_j))} \dx \\
  &+ \int_{x_i}^{a} \frac{\lambda}{2} \exp(-\lambda(x-x_i)) \p{1 -
  \frac{1}{2} \exp(-\lambda(x-x_j))} \dx \\
  &= \frac{1}{2} - \p{\frac{3}{8} + \frac{\lambda}{4} (x_i-x_j)}
  \exp(-\lambda(x_i - x_j)) + \frac{1}{2} (1 - \exp(-\lambda(a-x_i))) -
  \frac{\lambda}{4} \int_{x_i}^{a} \exp(-\lambda(2x-x_i-x_j)) \dx \\
  &= 1 - \p{\frac{3}{8} + \frac{\lambda}{4} (x_i-x_j)}
  \exp(-\lambda(x_i - x_j)) - \frac{1}{2} \exp(-\lambda(a-x_i)) \\
  &+ \frac{1}{8} \exp(\lambda(x_i + x_j)) (\exp(-2\lambda a) - \exp(-2\lambda
  x_i)) \\
  &= 1 - \p{\frac{1}{2} + \frac{\lambda}{4} (x_i-x_j)} \exp(-\lambda(x_i -
  x_j)) - \frac{1}{2} \exp(-\lambda(a-x_i)) + \frac{1}{8} \exp(-\lambda
  (2a-x_i-x_j))
\end{align*}
The denominator is
\begin{align*}
  &\p{1 - \frac{1}{2} \exp(-\lambda(t-x_i))} \p{1 - \frac{1}{2}
  \exp(-\lambda(t-x_j))} \\
  &= 1 - \frac{1}{2} \exp(-\lambda(a-x_i)) -
  \frac{1}{2} \exp(-\lambda(a-x_j)) + \frac{1}{4}
  \exp(-\lambda(2a-x_i-x_j))
\end{align*}
Thus,
\begin{align*}
  \Pr[&X_i > X_j \given X_i < a, X_j < a] \\
  &= \frac{1 - \p{\frac{1}{2} + \frac{\lambda}{4} (x_i-x_j)}
    \exp(-\lambda(x_i - x_j)) - \frac{1}{2} \exp(-\lambda(a-x_i)) +
    \frac{1}{8} \exp(-\lambda (2a-x_i-x_j))}{1 - \frac{1}{2}
    \exp(-\lambda(a-x_i)) - \frac{1}{2} \exp(-\lambda(a-x_j)) + \frac{1}{4}
  \exp(-\lambda(2a-x_i-x_j))} \\
  &\propto \frac{8 - \p{4 + 2\lambda (x_i-x_j)}
    \exp(-\lambda(x_i - x_j)) - 4 \exp(-\lambda(a-x_i)) + \exp(-\lambda
    (2a-x_i-x_j))}{4 - 2 \exp(-\lambda(a-x_i)) - 2 \exp(-\lambda(a-x_j))
    + \exp(-\lambda(2a-x_i-x_j))}
\end{align*}
We're interested in
\begin{align*}
  &\frac{d}{da} \Pr[X_i > X_j \given X_i < a, X_j < a] > 0 \\
  \Longleftrightarrow &
  \p{4 - 2 \exp(-\lambda(a-x_i)) - 2
  \exp(-\lambda(a-x_j)) + \exp(-\lambda(2a - x_i - x_j} \\
  &\cdot
  \p{4\lambda \exp(-\lambda(a-x_i)) - 2\lambda
  \exp(-\lambda(2a - x_i - x_j))} \\
  &>
  \p{8 - 4 \exp(-\lambda(a-x_i)) + \exp(-\lambda(2a-x_i-x_j)) - \p{4 +
  2\lambda (x_i-x_j)} \exp(-\lambda(x_i-x_j))} \\
  &\cdot
  \p{2\lambda \exp(-\lambda(a-x_i)) + 2\lambda
  \exp(-\lambda(a-x_j)) - 2\lambda \exp(-\lambda(2a-x_i-x_j))} \\
  \Longleftrightarrow &
  16 \exp(-\lambda(a-x_i)) - 8
  \exp(-\lambda(2a-x_i-x_j)) - 8 \exp(-2\lambda(a-x_i)) +
  4 \exp(-\lambda(3a-2x_i-x_j)) \\
  &-
  8\exp(-\lambda(2a-x_i-x_j)) + 4\exp(-\lambda(3a-x_i-2x_j)) +
  4\exp(-\lambda(3a-2x_i-x_j)) \\
  &- 2\exp(-2\lambda(2a-x_i-x_j)) \\
  &>
  16\exp(-\lambda(a-x_i)) + 16\exp(-\lambda(a-x_j)) -
  16\exp(-\lambda(2a-x_i-x_j)) \\
  &- 8\exp(-2\lambda(a-x_i)) - 8\exp(-\lambda(2a-x_i-x_j)) +
  8\exp(-\lambda(3a-2x_i-x_j)) \\
  &+ 2\exp(-\lambda(3a-2x_i-x_j)) + 2\exp(-\lambda(3a-x_i-2x_j)) -
  2\exp(-2\lambda(2a-x_i-x_j)) \\
  &-2(4+2\lambda(x_i-x_j))\exp(-\lambda(a-x_j))
  -2(4+2\lambda(x_i-x_j))\exp(-\lambda(a+x_i-2x_j)) \\
  &+ 2(4+2\lambda(x_i-x_j))\exp(-2\lambda(a-x_j)) \\
  \Longleftrightarrow &
  \exp(-\lambda(3a-x_i-2x_j)) \\
  &>
  8\exp(-\lambda(a-x_j))
  - 4\exp(-\lambda(2a-x_i-x_j))
  + \exp(-\lambda(3a-2x_i-x_j)) \\
  &-(4+2\lambda(x_i-x_j))\exp(-\lambda(a-x_j))
  -(4+2\lambda(x_i-x_j))\exp(-\lambda(a+x_i-2x_j)) \\
  &+ (4+2\lambda(x_i-x_j))\exp(-2\lambda(a-x_j)) \\
  \Longleftrightarrow &
  \exp(-\lambda(2a-x_i-x_j)) \\
  &>
  8
  - 4\exp(-\lambda(a-x_i))
  + \exp(-\lambda(2a-2x_i)) \\
  &-(4+2\lambda(x_i-x_j))
  -(4+2\lambda(x_i-x_j))\exp(-\lambda(x_i-x_j))
  + (4+2\lambda(x_i-x_j))\exp(-\lambda(a-x_j)) \\
  \Longleftrightarrow &
  \exp(-\lambda(2a-x_i-x_j))
  -8
  + 4\exp(-\lambda(a-x_i))
  - \exp(-2\lambda(a-x_i)) \\
  &+(4+2\lambda(x_i-x_j))(1+\exp(-\lambda(x_i-x_j)))
  - (4+2\lambda(x_i-x_j))\exp(-\lambda(a-x_j)) \\
  &> 0 \numberthis \label{eq:separable}
\end{align*}
Note that for any $z \ge 0$, we have
\begin{align*}
  (4+2z)(1+e^{-z}) - 8 \ge 0
  &\Longleftrightarrow
  (2+z)(1+e^{-z}) \ge 4 \\
  &\Longleftrightarrow
  z + 2e^{-z} + ze^{-z} \ge 2
\end{align*}
For $z = 0$, this holds with equality, and the left hand side is increasing
since
\begin{align*}
  \frac{d}{dx} z + 2e^{-z} + ze^{-z} \ge 0
  &\Longleftrightarrow 1 - 2e^{-z} + e^{-z} - ze^{-z} \ge 0 \\
  &\Longleftrightarrow 1 \ge e^{-z} + ze^{-z} \\
  &\Longleftrightarrow \frac{1}{1+z} \ge e^{-z} \\
  &\Longleftrightarrow 1+z \le e^{z}
\end{align*}
Therefore, choosing $z = \lambda(x_i-x_j)$ and plugging back
to~\eqref{eq:separable},
we have
\begin{align*}
  &\exp(-\lambda(2a-x_i-x_j))
  -8
  + 4\exp(-\lambda(a-x_i))
  - \exp(-2\lambda(a-x_i)) \\
  &+(4+2\lambda(x_i-x_j))(1+\exp(-\lambda(x_i-x_j)))
  - (4+2\lambda(x_i-x_j))\exp(-\lambda(a-x_j))
  > 0 \\
  \Longleftarrow&
  \exp(-\lambda(2a-x_i-x_j))
  + 4\exp(-\lambda(a-x_i))
  - \exp(-2\lambda(a-x_i))
  - (4+2\lambda(x_i-x_j))\exp(-\lambda(a-x_j))
  > 0 \\
  \Longleftrightarrow&
  \exp(-\lambda(a-x_j))
  + 4
  - \exp(-\lambda(a-x_i))
  - (4+2\lambda(x_i-x_j))\exp(-\lambda(x_i-x_j))
  > 0 \\
  \Longleftrightarrow&
  4(1-\exp(-\lambda(x_i-x_j))) + \exp(-\lambda(a-x_i))
  (\exp(-\lambda(x_i-x_j))-1)
  - 2\lambda(x_i-x_j)\exp(-\lambda(x_i-x_j))
  > 0 \\
  \Longleftrightarrow&
  (4-\exp(-\lambda(a-x_i))) (1-\exp(-\lambda(x_i-x_j)))
  - 2\lambda(x_i-x_j)\exp(-\lambda(x_i-x_j))
  > 0 \\
  \Longleftarrow&
  3 (1-\exp(-\lambda(x_i-x_j)))
  - 2\lambda(x_i-x_j)\exp(-\lambda(x_i-x_j))
  > 0 \tag{$\exp(-\lambda(a-x_i)) < 1$}
\end{align*}
Again letting $z = \lambda(x_i-x_j)$, this is true if and only if
\begin{align*}
  3(1-e^{-z}) > 2ze^{-z}
  &\Longleftrightarrow 3(e^z-1) > 2z \\
  &\Longleftrightarrow 3e^z > 3+2z
\end{align*}
which is true because $e^z > 1+z$ for $z > 0$. This completes the proof for Case
3.

As a result, we have that 
\begin{equation*}
  \frac{d}{da} \Pr[X_i > X_j \given X_i < a, X_j < a] \ge 0
\end{equation*}
for all $a$, with strict inequality for some $a$, which proves the theorem.
\end{proof}

\begin{theorem}
  For any $a \in \mathbb R$ and $X_i = x_i + \sigma \varepsilon_i$ where
  $\varepsilon_i \sim \mc N(0, 1)$,
  \begin{equation*}
    \frac{d}{da} \Pr[X_i > X_j \given X_i < a, X_j < a] > 0.
  \end{equation*}
  \label{thm:gaussian_deriv}
\end{theorem}
\begin{proof}
  
Assume $\sigma = 1/\sqrt{2}$. This is without loss of generality
because for any instance with arbitrary $\sigma'$, there is an instance with
$\sigma = 1/\sqrt{2}$ that yields the same distribution over rankings (simply by
scaling all item values by $\sigma/\sigma'$). First, we have
\begin{align*}
  \Pr[X_i > X_j \given X_i < a, X_j < a]
  &= \frac{\int_{-\infty}^a \Pr[X_i = x] \Pr[X_j < x] \dx}{Pr[X_i < a]
  \Pr[X_j < a]} \\
  &= \frac{\int_{-\infty}^a \exp(-(x-x_i)^2)/\sqrt{\pi} \cdot
  (1+\erf(x-x_j))/2 \dx}{(1+\erf(a-x_i))/2 \cdot (1+\erf(a-x_j))/2} \\
  &= \frac{2}{\sqrt{\pi}} \frac{\int_{-\infty}^a \exp(-(x-x_i)^2)
  (1+\erf(x-x_j)) \dx}{(1+\erf(a-x_i)) \cdot (1+\erf(a-x_j))}
\end{align*}
The derivative with respect to $a$ is positive if and only if
\begin{align*}
  &(1+\erf(a-x_i)) (1+\erf(a-x_j)) \exp(-(a-x_i)^2)
  (1+\erf(a-x_j)) \\
  &> \int_{-\infty}^a \exp(-(x-x_i)^2) (1+\erf(x-x_j)) \dx \\
  &\cdot \frac{2}{\sqrt{\pi}}\p{(1 + \erf(a - x_i)) \exp(-(a-x_j)^2) + (1 +
  \erf(a - x_j)) \exp(-(a-x_i)^2)} \numberthis \label{eq:gaus_cond}
\end{align*}
Let $t = a-x_i$ and $\delta = x_i-x_j$. Then, using the fact that
\begin{align*}
  \int_{-\infty}^a \exp(-(x-x_i)^2) (1+\erf(x-x_j)) \dx
                   &= \int_{-\infty}^{a-x_i} \exp(-x^2) \dx + \int_{-\infty}^{a-x_i}
  \exp(-x^2) \erf(x+\delta) \dx \\
  &= \frac{\sqrt{\pi}}{2} (1 + \erf(a-x_i)) + \int_{-\infty}^{a-x_i}
  \exp(-x^2) \erf(x+\delta) \dx,
\end{align*}
\eqref{eq:gaus_cond} becomes
\begin{align*}
  \frac{\sqrt{\pi}}{2} \cdot \frac{(1+\erf(t))(1+\erf(t+\delta))^2
  \exp(-t^2)}{(1+\erf(t))\exp(-(t+\delta)^2) + (1+\erf(t+\delta))\exp(-t^2)}
  > \frac{\sqrt{\pi}}{2} (1+\erf(t)) + \int_{-\infty}^t \exp(-x^2)
  \erf(x+\delta) \dx \\
  \Longleftrightarrow
  \frac{(1+\erf(t))(1+\erf(t+\delta))^2
  \exp(-t^2)}{(1+\erf(t))\exp(-(t+\delta)^2) + (1+\erf(t+\delta))\exp(-t^2)}
  - (1+\erf(t)) - \frac{2}{\sqrt{\pi}} \int_{-\infty}^t \exp(-x^2)
  \erf(x+\delta) \dx > 0
  \numberthis \label{eq:deriv1}
\end{align*}
To show that this is true, we will use the fact that $f(t) > 0$ whenever the
following conditions are met:
\begin{enumerate}
  \item $f(t)$ is continuous and differentiable everywhere
  \item $\lim_{t \to -\infty} f(t) = 0$
  \item $\frac{d}{dt} f(t) > 0$
\end{enumerate}
We'll show that these conditions hold for the LHS of~\eqref{eq:deriv1}.
\begin{align*}
  &\lim_{t \to -\infty}
   \frac{(1+\erf(t))(1+\erf(t+\delta))^2
  \exp(-t^2)}{(1+\erf(t))\exp(-(t+\delta)^2) + (1+\erf(t+\delta))\exp(-t^2)} -
  (1+\erf(t)) - \frac{2}{\sqrt{\pi}} \int_{-\infty}^t \exp(-x^2)
  \erf(x+\delta) \dx \\
  =&
  \lim_{t \to -\infty}
  \frac{(1+\erf(t))(1+\erf(t+\delta))^2
  \exp(-t^2)}{(1+\erf(t))\exp(-(t+\delta)^2) + (1+\erf(t+\delta))\exp(-t^2)}
  \numberthis \label{eq:deriv1lim}
\end{align*}
Observe that both the numerator and denominator of~\eqref{eq:deriv1lim} are
positive, so this limit must be at least 0. We can upper bound it by
\begin{align*}
  \lim_{t \to -\infty}
  \frac{(1+\erf(t))(1+\erf(t+\delta))^2
  \exp(-t^2)}{(1+\erf(t))\exp(-(t+\delta)^2) + (1+\erf(t+\delta))\exp(-t^2)}
  &\le
  \lim_{t \to -\infty}
  \frac{(1+\erf(t))(1+\erf(t+\delta))^2
  \exp(-t^2)}{(1+\erf(t+\delta))\exp(-t^2)} \\
  &=
  \lim_{t \to -\infty}
  (1+\erf(t))(1+\erf(t+\delta)) \\
  &= 0
\end{align*}
Thus, the limit is 0. Now, we must show that the derivative is positive. The
derivative is
\begin{align*}
  &\frac{d}{dt}
  \b{\frac{(1+\erf(t))(1+\erf(t+\delta))^2
    \exp(-t^2)}{(1+\erf(t))\exp(-(t+\delta)^2) + (1+\erf(t+\delta))\exp(-t^2)} -
    (1+\erf(t)) - \frac{2}{\sqrt{\pi}} \int_{-\infty}^t \exp(-x^2)
  \erf(x+\delta) \dx} \\
  &=
  \frac{d}{dt}
  \b{\frac{(1+\erf(t))(1+\erf(t+\delta))^2
  \exp(-t^2)}{(1+\erf(t))\exp(-(t+\delta)^2) + (1+\erf(t+\delta))\exp(-t^2)}} -
  \frac{2}{\sqrt{\pi}} \exp(-t^2) - \frac{2}{\sqrt{\pi}} \exp(-t^2)
  \erf(t+\delta) \numberthis \label{eq:deriv1deriv}
\end{align*}
Taking this derivative and factoring out
\begin{align*}
  \frac{2(1+\erf(t))(1+\erf(t+\delta))\exp(4t^2)}{\sqrt{\pi}
    \left(\left(\operatorname{erf}{\left(t \right)} + 1\right) e^{t^{2}} +
      \left(\operatorname{erf}{\left(\delta + t \right)} + 1\right)
  e^{\left(\delta + t\right)^{2}}\right)^{2}},
\end{align*}
we get that~\eqref{eq:deriv1deriv} is positive if and only if
\begin{align*}
  &\delta \sqrt{\pi} \exp((t+\delta)^2) (1+\erf(t)) (1+\erf(t+\delta)) -
  \exp(2\delta t + t^2) (1+\erf(t+\delta)) + (1+\erf(t)) > 0 \\
  \Longleftrightarrow &
  \delta \sqrt{\pi} \exp((t+\delta)^2) (1+\erf(t)) +
  \frac{1+\erf(t)}{1+\erf(t+\delta)} - \exp(2\delta t + t^2) > 0 \\
  \Longleftrightarrow &
  \delta \sqrt{\pi} \exp(t^2) (1+\erf(t)) +
  \exp(-2\delta t - t^2) \frac{1+\erf(t)}{1+\erf(t+\delta)} - 1 > 0 \\
  \Longleftrightarrow &
  (1+\erf(t)) \b{\delta \sqrt{\pi} \exp(t^2) + \frac{\exp(-2\delta t -
  \delta^2)}{1+\erf(t+\delta)}} - 1 > 0 \\
  \Longleftrightarrow &
  \frac{1+\erf(t)}{\exp(-t^2)} \b{\delta \sqrt{\pi} +
  \frac{\exp(-(t+\delta)^2)}{1+\erf(t+\delta)}} - 1 > 0
  \numberthis \label{eq:deriv2_cond}
\end{align*}
Define
\begin{equation*}
  g(t) \triangleq \frac{1+\erf(t)}{\exp(-t^2)}.
\end{equation*}
Then,~\eqref{eq:deriv2_cond} is
\begin{align*}
  &g(t) \b{\delta \sqrt{\pi} + \frac{1}{g(t+\delta)}} - 1 > 0 \\
  \Longleftrightarrow &
  \frac{1}{g(t)} - \frac{1}{g(t+\delta)} < \delta \sqrt{\pi}
\end{align*}
By the Mean Value Theorem,
\begin{equation*}
  \frac{1}{g(t)} - \frac{1}{g(t+\delta)} = -\delta \left. \frac{d}{dt}
    \frac{1}{g(t)} \right|_{t=t^*}
\end{equation*}
for some $t \leq t^* \leq t+\delta$. Thus, it suffices to show that 
\begin{equation}
  \label{eq:deriv_bound}
  \frac{d}{dt} \frac{1}{g(t)} > -\sqrt{\pi}
\end{equation}
for all $t$. To do this, consider Mills Ratio~\cite{mills1926table}
\begin{equation*}
  R(t) \triangleq \exp(t^2/2) \int_t^{\infty} \exp(-x^2/2) \dx.
\end{equation*}
Note that this is quite similar in functional form to $g(t)$, and with some
manipulation, we can relate the two:
\begin{align*}
  R(t) &= \exp(t^2/2) \int_t^{\infty} \exp(-x^2/2) \dx \\
  R(\sqrt{2} t) &= \exp(t^2) \int_{\sqrt{2} t}^{\infty} \exp(-x^2/2) \dx \\
                &= \sqrt{2} \exp(t^2) \int_t^\infty \exp(-x^2) \dx \\
                &= \sqrt{2} \exp(t^2) \int_{-\infty}^{-t} \exp(-x^2) \dx
                \tag{$\exp(-x^2)$ is symmetric} \\
  R(-\sqrt{2} t) &= \sqrt{2} \exp(t^2) \int_{-\infty}^{t} \exp(-x^2) \dx \\
                 &= \sqrt{2} \exp(t^2) \cdot \frac{\sqrt{\pi}}{2} (1+\erf(t)) \\
                 &= \sqrt{\frac{\pi}{2}} \p{\frac{1+\erf(t)}{\exp(-t^2)}} \\
  R(-\sqrt{2} t) &= \sqrt{\frac{\pi}{2}} g(t)
\end{align*}
Sampford~\cite[Eq. (3)]{sampford1953some} proved that $\frac{d}{dt}
\frac{1}{R(t)} < 1$ for any $t$. Thus,
\begin{align*}
  \frac{d}{dt} \frac{1}{g(t)}
  = \frac{d}{dt} \frac{1}{\sqrt{\frac{2}{\pi}} R(-\sqrt{2} t)}
  = \sqrt{\frac{\pi}{2}} \frac{d}{dt} \frac{1}{R(-\sqrt{2} t)} >
  \sqrt{\frac{\pi}{2}} \cdot 1 \cdot -\sqrt{2} = -\sqrt{\pi},
\end{align*}
which proves~\eqref{eq:deriv_bound} and completes the proof.
\end{proof}

\section{Verifying that the Mallows Model Satisfies Definition~\ref{def:family}}
\label{app:mallows_family}
\begin{theorem}
  The family of distributions $\mc F_\theta$ produced by the Mallows Model with
  Kendall tau distance with $\theta = \phi - 1$ satisfies the conditions of
  Definition~\ref{def:family}. 
\end{theorem}
\begin{proof}
  We must show that $\mc F_\theta$ satisfies the differentiability, asymptotic
  optimality, and monotonicity conditions of Definition~\ref{def:family}.

  \textbf{Differentiability:} Let $\Pi$ be the set of all permutations on $n$
  candidates. The probability of a realizing a particular permutation $\pi$
  under the Mallows model is
  \begin{equation*}
    \Pr_\theta[\pi] = \frac{\phi^{-d(\pi, \pi^*)}}{\sum_{\pi' \in \Pi}
    \phi^{-d(\pi', \pi^*)}}
  \end{equation*}
  Both the numerator and denominator are differentiable with respect to $\theta
  = \phi-1$, so $\Pr_\theta[\pi]$ is differentiable with respect to $\theta$.

  \textbf{Asymptotic optimality:} For the correct ranking $\pi^*$,
  \begin{equation*}
    \Pr_\theta[\pi^*] = \frac{1}{Z},
  \end{equation*}
  where the normalizing constant $Z$ is
  \begin{equation*}
    Z = \sum_{\pi \in \Pi} \phi^{-d(\pi, \pi^*)}
  \end{equation*}
  In the limit,
  \begin{align*}
    \lim_{\theta \to \infty} Z
    &= \lim_{\phi \to \infty} Z \\
    &= \lim_{\phi \to \infty}  \sum_{\pi \in \Pi} \phi^{-d(\pi, \pi^*)} \\
    &= \lim_{\phi \to \infty} 1 + \sum_{\pi \ne \pi^* \in \Pi} \phi^{-d(\pi, \pi^*)} \\
    &= 1 + \sum_{\pi \ne \pi^* \in \Pi} \lim_{\phi \to \infty} \phi^{-d(\pi, \pi^*)} \\
    &= 1
  \end{align*}
  because for any $\pi \ne \pi^*$, $d(\pi, \pi^*) \ge 1$. Therefore,
  \begin{equation*}
    \lim_{\theta \to \infty} \Pr_\theta[\pi^*] = \lim_{\theta \to \infty}
    \frac{1}{Z} = 1
  \end{equation*}

  \textbf{Monotonicity:}
  We must show that for any $S \subset \mathbf{x}$, if $\pi_1^{(-S)}$ denotes
  the value of the top-ranked candidate according to $\pi$ excluding candidates
  in $S$,
  \begin{equation*}
    \mathbb{E}_{\mc F_{\theta'}}\b{\pi_1^{(-S)}} \ge \mathbb{E}_{\mc
    F_{\theta}}\b{\pi_1^{(-S)}}.
  \end{equation*}
  For any $i \notin S$, let $j$ be the smallest index such that $j > i$ and $j
  \notin S$. Consider any $\pi$ such that $\pi_1^{(-S)} = x_j$. Then, swapping
  $i$ and $j$ yields a permutation $\hat \pi$ such that $\hat \pi_1^{(-S)} =
  x_i$. Moreover,
  \begin{align*}
    \Pr[\hat \pi] = \Pr[\pi] \cdot \phi^{\inv(\pi) - \inv(\hat \pi)}.
  \end{align*}
  Since $i < j$, $\inv(\pi) - \inv(\hat \pi) \ge 1$. Finally, note that swapping
  $i$ and $j$ is a bijection between $\{\pi : \pi_1^{(-S)} = x_j\}$ and $\{\pi :
  \pi_1^{(-S)} = x_i\}$. Thus,
  \begin{align*}
    \frac{\Pr[\pi_1^{(-S)} = x_i]}{\Pr[\pi_1^{(-S)} = x_j]}
    &= \sum_{\pi : \pi_1^{(-S)} = x_j} \frac{\Pr[\pi]}{\Pr[\pi_1^{(-S)} = x_j]}
    \cdot \phi^{\inv(\pi) - \inv(\hat \pi)}
  \end{align*}
  Note that the terms $\frac{\Pr[\pi]}{\Pr[\pi_1^{(-S)} = x_j]}$ sum to 1, so
  this is sum is some polynomial in $\phi$ with nonnegative weights and integer
  powers of $\phi$. As a result, it must have a positive derivative with respect
  to $\phi$, i.e., for $i < j$,
  \begin{align*}
    \frac{d}{d\phi} \frac{\Pr[\pi_1^{(-S)} = x_i]}{\Pr[\pi_1^{(-S)} = x_j]} > 0
  \end{align*}
  Let $\phi' > \phi$. Then,
  \begin{align*}
    \frac{\Pr_\phi[\pi_1^{(-S)} = x_i]}{\Pr_\phi[\pi_1^{(-S)} = x_j]}
    &< \frac{\Pr_{\phi'}[\pi_1^{(-S)} = x_i]}{\Pr_{\phi'}[\pi_1^{(-S)} = x_j]}
  \end{align*}
  Rearranging,
  \begin{equation}
    \label{eq:phi_likelihood}
    \frac{\Pr_\phi[\pi_1^{(-S)} = x_i]}{\Pr_{\phi'}[\pi_1^{(-S)} = x_i]}
    < \frac{\Pr_\phi[\pi_1^{(-S)} = x_j]}{\Pr_{\phi'}[\pi_1^{(-S)} = x_j]}
  \end{equation}
  For $\theta' - \phi' - 1$ and $\theta = \phi - 1$,
  \begin{align*}
    \mathbb{E}_{\mc F_\theta} \b{\pi_1^{(-S)}} &= \sum_{i \notin S} \Pr_\phi
    \b{\pi_1^{-(S)} = x_i} x_i \\
    \mathbb{E}_{\mc F_{\theta'}} \b{\pi_1^{(-S)}} &= \sum_{i \notin S}
    \Pr_{\phi'} \b{\pi_1^{-(S)} = x_i} x_i
  \end{align*}
  By Lemma~\ref{lem:prob_sum},
  \begin{align*}
    \mathbb{E}_{\mc F_{\theta'}} \b{\pi_1^{(-S)}} > \mathbb{E}_{\mc F_\theta}
    \b{\pi_1^{(-S)}},
  \end{align*}
  which completes the proof. Note that we apply Lemma~\ref{lem:prob_sum}
  indexing backwards from $n$ to 1, ignoring elements in $S$, with $p_i =
  \Pr_\phi \b{\pi_1^{-(S)} = x_i}$ and $q_i = \Pr_{\phi'} \b{\pi_1^{-(S)} =
  x_i}$. \eqref{eq:phi_likelihood} provides the condition that $p_i/q_i$ is
  decreasing (as $i$ decreases, since we are indexing backwards).
\end{proof}

\section{Proof of Theorem~\ref{thm:mallows}}
\label{app:mallows_proof}
\subsection{Verifying Definition~\ref{def:first}}
We must show that when $\pi, \tau \sim \mc F_\theta$,
\begin{equation}
  \label{eq:p1p2_cond}
  \E{\pi_1 - \pi_2 \given \pi_1 \ne \tau_1} > 0.
\end{equation}
We begin by expanding:
\begin{align*}
  \E{\pi_1 - \pi_2 \given \pi_1 \ne \tau_1}
  &= \sum_{i=1}^n \sum_{j=1}^n (x_i - x_j) \Pr[\pi_1 = x_i \cap \pi_2 = x_j
  \given \pi_1 \ne \tau_1] \\
  &= \sum_{i=1}^{n-1} \sum_{j > i} (x_i - x_j) \p{\Pr[\pi_1 = x_i \cap \pi_2 =
    x_j \given \pi_1 \ne \tau_1] - \Pr[\pi_1 = x_j \cap \pi_2 = x_i \given \pi_1
  \ne \tau_1]}
\end{align*}
Since $x_i > x_j$ for $i < j$, it suffices to show that for all $i < j$,
\begin{equation}
  \label{eq:pr_cond2}
  \Pr[\pi_1 = x_i \cap \pi_2 = x_j \given \pi_1 \ne \tau_1] \ge \Pr[\pi_1 = x_j
  \cap \pi_2 = x_i \given \pi_1 \ne \tau_1],
\end{equation}
and that this holds strictly for some $i < j$. We simplify~\eqref{eq:pr_cond2}
as follows:
\begin{align*}
  &\Pr[\pi_1 = x_i \cap \pi_2 = x_j \given \pi_1 \ne \tau_1]
  > \Pr[\pi_1 = x_j \cap \pi_2 = x_i \given \pi_1 \ne \tau_1] \\
  &\Longleftrightarrow \frac{\Pr[\pi_1 = x_i \cap \pi_2 = x_j \cap \pi_1 \ne
  \tau_1]}{\Pr[\pi_1 \ne \tau_1]}
  > \frac{\Pr[\pi_1 = x_j \cap \pi_2 = x_i \cap \pi_1 \ne \tau_1]}{\Pr[\pi_1 \ne
  \tau_1]} \\
  &\Longleftrightarrow \Pr[\pi_1 = x_i \cap \pi_2 = x_j \cap \pi_1 \ne \tau_1]
  > \Pr[\pi_1 = x_j \cap \pi_2 = x_i \cap \pi_1 \ne \tau_1] \\
  &\Longleftrightarrow \Pr[\pi_1 = x_i \cap \pi_2 = x_j \cap \tau_1 \ne x_i]
  > \Pr[\pi_1 = x_j \cap \pi_2 = x_i \cap \tau_1 \ne x_j] \\
  &\Longleftrightarrow \Pr[\pi_1 = x_i \cap \pi_2 = x_j] \Pr[\tau_1 \ne x_i]
  > \Pr[\pi_1 = x_j \cap \pi_2 = x_i] \Pr[\tau_1 \ne x_j]
  \numberthis \label{eq:mallow_cond}
\end{align*}
We can simplify~\eqref{eq:mallow_cond} using Lemmas~\ref{lem:ABBA}
and~\ref{lem:p1xj}. Let $|i-j|$ denote the difference in rank between $x_i$ and
$x_j$.
\begin{align*}
  &\Pr[\pi_1 = x_i \cap \pi_2 = x_j] \Pr[\tau_1 \ne x_i] - \Pr[\pi_1 = x_j \cap
  \pi_2 = x_i] \Pr[\tau_1 \ne x_j] \\
  &= \Pr[\pi_1 = x_i \cap \pi_2 = x_j] (1-\Pr[\tau_1 = x_i])
  - \phi^{-1} \Pr[\pi_1 = x_i \cap \pi_2 = x_j] (1-\Pr[\tau_1 = x_j]) \\
  &= \Pr[\pi_1 = x_i \cap \pi_2 = x_j] (1-\Pr[\tau_1 = x_i])
  - \phi^{-1} \Pr[\pi_1 = x_i \cap \pi_2 = x_j] (1-\phi^{-|i-j|}\Pr[\tau_1 =
  x_i]) \\
  &= \Pr[\pi_1 = x_i \cap \pi_2 = x_j] (1-\Pr[\tau_1 = x_i]
  - \phi^{-1} -\phi^{-|i-j|-1}\Pr[\tau_1 = x_i]))
\end{align*}
This is positive if and only if
\begin{align*}
  &1-\Pr[\tau_1 = x_i] - \phi^{-1} -\phi^{-|i-j|-1}\Pr[\tau_1 = x_i] > 0 \\
  &\Longleftrightarrow \Pr[\tau_1 = x_i] \p{1 - \phi^{-|i-j|-1}} < 1 - \phi^{-1}
  \\
  &\Longleftrightarrow \Pr[\tau_1 = x_i] < \frac{1 - \phi^{-1}}{1 -
  \phi^{-|i-j|-1}} \\
  &\Longleftrightarrow \frac{1 - \phi^{-1}}{\phi^{i-1}(1-\phi^{-n})} < \frac{1 -
  \phi^{-1}}{1 - \phi^{-|i-j|-1}} \\
  &\Longleftrightarrow {\phi^{i-1}(1-\phi^{-n})} > {1 - \phi^{-|i-j|-1}}
\end{align*}
This is weakly true for any $i<j$ because $\phi^{i-1} \ge 1$ and $|i-j| + 1 \le
n$, and it is strictly true for any $i,j$ other than $1$ and $n$. Thus, $
\E{\pi_1 - \pi_2 \given \pi_1 \ne \tau_1} > 0$.

\subsection{Verifying Definition~\ref{def:weaker}}

Recall that Definition~\ref{def:weaker} is equivalent to $U_{AH}(\theta_A,
\theta_H) < U_{HH}(\theta_A, \theta_H)$ for $\theta_A > \theta_H$. Let $\tau$ be
the algorithmic ranking, and let $\pi$ be a ranking from a human evaluator.
Recall that $U_H(\theta_A, \theta_H) = \E{\pi_1}$. Throughout this proof, we
will drop the $(\theta_A, \theta_H)$ notation and simply write $U_H$, $U_{AH}$,
and $U_{HH}$.
\begin{align*}
  U_{AH}
  &= \sum_{i=1}^n (\Pr[\pi_1 = x_i \cap \tau_1 \ne x_i] + \Pr[\pi_2 = x_i \cap
  \pi_1 = \tau_1]) x_i \\
  &= \sum_{i=1}^n \Pr[\pi_1 = x_i \cap \tau_1 \ne x_i] x_i + \sum_{i=1}^n
  \Pr[\pi_2 = x_i \cap \pi_1 = \tau_1] x_i \\
  &= \sum_{i=1}^n \p{Pr[\pi_1 = x_i] - \Pr[\pi_1 = x_i \cap \tau_1 = x_i]} x_i
  + \sum_{i=1}^n \sum_{j \ne i} \Pr[\pi_1 = x_j \cap \tau_1 = x_j \cap \pi_2 =
  x_i] x_i \\
  &= U_H - \sum_{i=1}^n \Pr[\pi_1 = x_i \cap \tau_1 = x_i] x_i + \sum_{i=1}^n
  \Pr[\pi_1 = x_i \cap \tau_1 = x_i] \E{\pi_2 \given \pi_1 = x_i \cap \tau_1 =
  x_i} \\
  &= U_H + \sum_{i=1}^n \Pr[\pi_1 = x_i] \Pr[\tau_1 = x_i] \p{\E{\pi_2 \given
  \pi_1 = x_i} - x_i}
\end{align*}
Similarly, because two human evaluators are independent,
\begin{align*}
  U_{HH} = U_H + \sum_{i=1}^n \Pr[\pi_1 = x_i]^2 \p{\E{\pi_2 \given \pi_1 =
  x_i} - x_i}.
\end{align*}
Let $V_{-i} = \E{\pi_2 \given \pi_1 = x_i}$. Note that conditioned on $\pi_1
= x_i$, the remaining elements of $\pi_1$ follow a Mallows model
distribution over
$n-1$ candidates. Because the Mallows model is value-independent, increasing any
item value increases the expected value of the top-ranked item (and in fact, the
item ranked at any position). Thus, $V_{-i}$ increases as $i$ increases (since
$x_i$, the value of the unavailable candidate, decreases). Moreover, $x_i$ is
strictly decreasing in $i$, so $V_{-i} - x_i$ is strictly increasing in $i$.
With this, we have
\begin{align*}
  U_{AH} - U_H &= \sum_{i=1}^n \Pr[\pi_1 = x_i] \Pr[\tau_1 = x_i] \p{V_{-i} -
  x_i} \\
  U_{HH} - U_H &= \sum_{i=1}^n \Pr[\pi_1 = x_i]^2 \p{V_{-i} -
  x_i}
\end{align*}
Let $C_A = \Pr[\pi_1 = \tau_1] = \sum_{i=1}^n \Pr[\pi_1 = x_i] \Pr[\tau_1 =
x_i]$, and similarly let $C_H = \sum_{i=1}^n \Pr[\pi_1 = x_i]^2$. $C_A > C_H$
by Lemma~\ref{lem:prob_sum} with $y_i' = \Pr[\pi_1 = x_{n-i+1}]$, $p_i' =
\Pr[\pi_1 = x_{n-i+1}]$ and $q_i' = \Pr[\tau_1 = x_{n-i+1}]$.

Let $p_i = \Pr[\pi_1' = i] \Pr[\pi_1 = i]/C_A$, $q_i = \Pr[\pi_1' = i]^2/C_H$, and
$y_i = V_{-i} - x_i$. Then, we have
\begin{align*}
  \frac{U_{AH} - U_H}{C_A} &= \sum_{i=1}^n p_i y_i \\
  \frac{U_{HH} - U_H}{C_H} &= \sum_{i=1}^n q_i y_i
\end{align*}
With $\phi_A = \theta_A + 1$ and $\phi_H = \theta_H + 1$,
\begin{align*}
  \frac{p_i}{q_i} = 
  \frac{C_H}{C_A} \cdot \frac{\frac{1 -
  \phi_A^{-1}}{\phi_A^{i-1}(1-\phi_A^{-n})}}{\frac{1 - \phi_H^{-1}}{\phi_H^{i-1}(1-\phi_H^{-n})}}
  &\propto \frac{\phi_H^{i-1}}{\phi_A^{i-1}},
\end{align*}
which is decreasing in $i$ since $\phi_H < \phi_A$. By Lemma~\ref{lem:prob_sum},
$\sum_{i=1}^n p_i y_i < \sum_{i=1}^n q_i y_i$. Finally, note that $U_{HH} - U_H
< 0$ by Lemma~\ref{lem:UH_UHH}, so
\begin{align*}
  \sum_{i=1}^n p_i y_i &< \sum_{i=1}^n q_i y_i \\
  \frac{U_{AH} - U_H}{C_A} &< \frac{U_{HH} - U_H}{C_H} \\
  \frac{C_H (U_{AH} - U_H)}{C_A} &< U_{HH} - U_H \\
  U_{AH} - U_H &< U_{HH} - U_H \tag{$C_A > C_H$, and $U_{HH} - U_H < 0$} \\
  U_{AH} &< U_{HH}
\end{align*}

\section{Supplementary Lemmas for the Mallows Model}

\begin{lemma}
  Let $\{y_i\}_{i=1}^n$, $\{p_i\}_{i=1}^n$, and $\{q_i\}_{i=1}^n$ be sequences
  such that
  \begin{itemize}
    \item $y_i$ is strictly increasing.
    \item $\sum_{i=1}^n p_i = \sum_{i=1}^n q_i = 1$.
    \item $\frac{p_i}{q_i}$ is decreasing.
  \end{itemize}
  Then, $\sum_{i=1}^n p_i y_i < \sum_{i=1}^n q_i y_i$.
  \label{lem:prob_sum}
\end{lemma}
\begin{proof}
  First, note that there exists $j$ such that $p_i > q_i$ for $i < j$ and $p_i
  \le q_i$ for $i \ge j$. To see this, let $j$ be the smallest index such that
  $p_j \le q_j$. Such a $j$ must exist because $p_i$ and $q_i$ both sum to 1, so
  it cannot be the case that $p_i > q_i$ for all $i$. This implies $p_i/q_i \le
  1$, and since $p_i/q_i$ is decreasing, $p_i \le q_i$ for $i \ge j$.

  Next, note that
  \begin{align*}
    0
    &= \sum_{i=1}^n (p_i - q_i) \\
    &= \sum_{i=1}^{j-1} (p_i - q_i) + \sum_{i=j}^n (p_i - q_i),
  \end{align*}
  meaning
  \begin{equation*}
    \sum_{i=1}^{j-1} (p_i - q_i) = \sum_{i=j}^n (q_i - p_i).
  \end{equation*}

  Using this choice of $j$, we can write
  \begin{align*}
    \sum_{i=1}^n p_i y_i - \sum_{i=1}^n q_i y_i
    &= \sum_{i=1}^n (p_i - q_i) y_i \\
    &= \sum_{i=1}^{j-1} (p_i - q_i) y_i - \sum_{i=j}^n (q_i - p_i) y_i \\
    &\le \sum_{i=1}^{j-1} (p_i - q_i) y_{j-1} - \sum_{i=j}^n (q_i - p_i) y_j \\
    &= \sum_{i=1}^{j-1} (p_i - q_i) y_{j-1} - \sum_{i=j}^n (q_i - p_i) y_j \\
    &= \sum_{i=1}^{j-1} (p_i - q_i) y_{j-1} - \sum_{i=1}^{j-1} (p_i - q_i) y_j
    \\
    &= \sum_{i=1}^{j-1} (p_i - q_i) (y_{j-1}-y_j) \\
    &< 0
  \end{align*}
\end{proof}

\begin{lemma}
  \label{lem:ABBA}
  For $x_i > x_j$,
  \begin{equation}
    \Pr[\pi_1 = x_i \cap \pi_2 = x_j] = \phi \Pr[\pi_1 = x_j \cap \pi_2 = x_i].
    \label{eq:ABBA}
  \end{equation}
\end{lemma}
\begin{proof}
  Let $\pi_{-ij}$ be a permutation of all of the candidates except $x_i$ and
  $x_j$. Then, we have
  \begin{align*}
    \Pr[\pi_1 = x_i \cap \pi_2 = x_j]
    &= \sum_{\pi_{-ij}} \Pr[\pi_1 = x_i \cap \pi_2 = x_j \given \pi_{-ij}]
    \Pr[\pi_{-ij}] \\
    &= \sum_{\pi_{-ij}} \phi \Pr[\pi_1 = x_j \cap \pi_2 = x_i \given \pi_{-ij}]
    \Pr[\pi_{-ij}] \\
    &= \phi \Pr[\pi_1 = x_j \cap \pi_2 = x_i]
  \end{align*}
  Intuitively, given that $x_i$ and $x_j$ are in the top 2 positions, $x_i$
  followed by $x_j$ is $\phi$
  times more likely than $x_j$ followed by $x_i$ regardless of the remainder of
  the permutation, and therefore, $x_i$ followed by $x_j$ is $\phi$ times more
  likely overall.
\end{proof}

\begin{lemma}
  \label{lem:p1xj}
  For $1 \le i \le n$,
  \begin{equation}
    \Pr[\pi_1 = x_i] = \frac{1 - \phi^{-1}}{\phi^{i-1} (1 - \phi^{-n})}.
    \label{eq:p1xi}
  \end{equation}
\end{lemma}
\begin{proof}
  Let $\pi_{-i}$ be a permutation over all items except $i$. Then,

  \begin{align*}
    \Pr[\pi_1 = x_i]
    &= \sum_{\pi_{-i}} \Pr[\pi_1 = x_i \given \pi_{-i}] \Pr[\pi_{-i}] \\
    &= \sum_{\pi_{-i}} \phi^{-(i-1)} \Pr[\pi_{-i}] \\
    &= \phi^{-(i-1)} \sum_{\pi_{-i}} \Pr[\pi_{-i}]
  \end{align*}

  Note that $\Pr[\pi_{-i}]$ doesn't depend on \textit{which} $n-1$ items are
  being ranked, so this term appears for any $i$. Moreover, $\sum_{i=1}^n
  \Pr[\pi_1 = x_i] = 1$. Therefore, we have

  \begin{equation*}
    \Pr[\pi_1 = x_i] \propto \phi^{-(i-1)}.
  \end{equation*}

  Normalizing, we get

  \begin{align*}
    \Pr[\pi_1 = x_i]
    &= \frac{\phi^{-(i-1)}}{\sum_{j=1}^n \phi^{-(j-1)}} \\
    &= \frac{\phi^{-(i-1)}}{\frac{1-\phi^{-n}}{1 - \phi^{-1}}} \\
    &= \frac{1 - \phi^{-1}}{\phi^{i-1}(1-\phi^{-n})}
  \end{align*}
  Intuitively, any permutation over $n-1$ items is equally likely regardless of
  what those items are, and inserting any item at the front of this permutation
  yields a likelihood proportional to the number of additional inversions this
  causes, which is equal to the item's position on the
  list.\footnote{Alternatively, we could prove this by showing that for any
    permutation with $i$ in front, the permutation in which $i$ and
    $i-1$ are swapped is $\phi$ times more likely, and thus, $i-1$ is
  $\phi$ times more likely to be in front than $i$.}
\end{proof}

\begin{lemma}
  \label{lem:UH_UHH}
  For the Mallows Model, $U_H(\theta_A, \theta_H) > U_{HH}(\theta_A, \theta_H)$.
\end{lemma}
\begin{proof}
  Intuitively, this is because selecting first is better than selecting second.
  To prove this, let $\pi$ and $\tau$ be ranking generated by independent human
  evaluators under the Mallows Model, i.e., $\pi, \tau \sim \mc F_{\theta_H}$.
  \begin{align*}
    U_H(\theta_A, \theta_H) - U_{HH}(\theta_A, \theta_H)
    &= \E{\pi_1} - \E{\tau_1 \cdot \ind{\pi_1 \ne \tau_1} + \tau_2 \cdot
    \ind{\pi_1 = \tau_1}} \\
    &= \E{(\pi_1 - \tau_2) \cdot \ind{\pi_1 = \tau_1}} \\
    &= \E{(\pi_1 - \pi_2) \cdot \ind{\pi_1 = \tau_1}}
  \end{align*}
  For any $i < j$, conditioned on $\pi_1 = \tau_1$, they are more likely to be
  correctly ordered than not:
  \begin{align*}
    \E{(\pi_1 - \pi_2) \cdot \ind{\pi_1 = \tau_1}}
    &= \sum_{i < j} \p{\Pr[\pi_1 = x_i \cap \tau_1 = x_i \cap \pi_2 = x_j] -
    \Pr[\pi_1 = x_j \cap \tau_1 = x_j \cap \pi_2 = x_i]} (x_i - x_j) \\
    &= \sum_{i < j} \p{\Pr[\pi_1 = x_i \cap \pi_2 = x_j] \Pr[\tau_1 = x_i] -
    \Pr[\pi_1 = x_j \cap \pi_2 = x_i] \Pr[\tau_1 = x_j]} (x_i - x_j) \\
    &> \sum_{i < j} \p{\Pr[\pi_1 = x_i \cap \pi_2 = x_j] \Pr[\tau_1 = x_j] -
    \Pr[\pi_1 = x_j \cap \pi_2 = x_i] \Pr[\tau_1 = x_j]} (x_i - x_j) \\
    &= \sum_{i < j} \p{\Pr[\pi_1 = x_i \cap \pi_2 = x_j] -
    \Pr[\pi_1 = x_j \cap \pi_2 = x_i]} (x_i - x_j) \\
    &\ge \sum_{i < j} \p{\phi_H\Pr[\pi_1 = x_j \cap \pi_2 = x_i] -
    \Pr[\pi_1 = x_j \cap \pi_2 = x_i]} (x_i - x_j) \\
    &> 0
  \end{align*}
\end{proof}

\subsection{Proof of Theorem~\ref{thm:H_opt}}
\label{app:H_opt_proof}
To prove this theorem, we make use of the following lemma.
\begin{lemma}
  \label{lem:monotone}
  Under the Mallows model, the probability that any two items $i < j$ are
  correctly ranked increases monotonically with the accuracy parameter $\phi$.
\end{lemma}
\def\Sij{S_{i \succ j}}
\def\Sji{S_{j \succ i}}
\def\Son{S_{1 \succ n}}
\def\Sno{S_{n \succ 1}}
\begin{proof}
  Let $\inv(\pi)$ be the number of inversions in a permutation $\pi$.
  Under the Mallows model, the probability of observing $\pi$ is proportional to
  $\phi^{-\inv(\pi)}$. Let $\Sij$ (resp. $\Sji$) be the set of
  permutations where $i$ is ranked before $j$ (resp. $j$ is ranked before $i$).
  Then, the probability $i$ is ranked before $j$ is
  \begin{align*}
    \Pr[i \succ j]
    &= \frac{\sum_{\pi \in \Sij} \phi^{-\inv(\pi)}}{\sum_{\pi \in \Sij}
    \phi^{-\inv(\pi)} + \sum_{\pi \in \Sji} \phi^{-\inv(\pi)}}.
  \end{align*}
  We will show that $\frac{d}{d\phi} \Pr[i \succ j] > 0$. Note that this is
  equivalent to showing $\frac{d}{d\phi} \frac{\Pr[i \succ j]}{\Pr[j \succ i]} >
  0$. Note that
  \begin{align*}
    \frac{\Pr[i \succ j]}{\Pr[j \succ i]}
    &= \frac{\sum_{\pi \in \Sij} \phi^{-\inv(\pi)}}{\sum_{\pi \in \Sji}
    \phi^{-\inv(\pi)}}.
  \end{align*}
  Let $\pi_{i:j}$ be the subsequence of $\pi$ containing elements $i$ through
  $j$. Then, we have
  \begin{align*}
    \frac{\Pr[i \succ j]}{\Pr[j \succ i]}
    &= \frac{\sum_{\pi \in \Sij} \phi^{-\inv(\pi)}}{\sum_{\pi \in \Sji}
    \phi^{-\inv(\pi)}} \\
    &= \frac{\sum_{\pi_{i:j} : \pi \in \Sij} \phi^{-\inv(\pi_{i:j})} \sum_{\pi':
        \pi'_{i:j} = \pi_{i:j}} \phi^{\inv(\pi_{i:j})-\inv(\pi')}}{\sum_{\pi_{i:j} :
      \pi \in \Sji} \phi^{-\inv(\pi_{i:j})} \sum_{\pi': \pi'_{i:j} = \pi_{i:j}}
    \phi^{\inv(\pi_{i:j})-\inv(\pi')}} \\
    &= \frac{\sum_{\pi_{i:j} : \pi \in \Sij}
      \phi^{-\inv(\pi_{i:j})}}{\sum_{\pi_{i:j} : \pi \in \Sji}
      \phi^{-\inv(\pi_{i:j})}}
  \end{align*}
  Intuitively, the term $\sum_{\pi': \pi'_{i:j} = \pi_{i:j}}
  \phi^{\inv(\pi_{i:j})-\inv(\pi')}$ does not depend on $\pi_{i:j}$ because for
  any $\pi_{i:j}$, if we fix the order and positions of the remaining elements,
  the number of inversions involving at least one element outside of $i:j$
  (i.e., $\inv(\pi') - \inv(\pi_{i:j})$) is a constant. For fixed $\pi_{i:j}$,
  there is a bijection between permutations $\pi': \pi'_{i:j} = \pi_{i:j}$ and a
  fixed order and position of the remaining elements (excluding $i:j$), meaning
  this sum does not depend on $\pi_{i:j}$. Thus, for the remainder of this
  proof, we can assume without loss of generality that $i=1$ and $j=n$. The
  quantity of interest becomes
  \begin{align*}
    \frac{\Pr[1 \succ n]}{\Pr[n \succ 1]}
    &= \frac{\sum_{\pi_{1:n} : \pi \in \Son}
      \phi^{-\inv(\pi_{1:n})}}{\sum_{\pi_{i:j} : \pi \in \Sno}
    \phi^{-\inv(\pi_{1:n})}} \\
    &= \frac{\sum_{\pi \in \Son} \phi^{-\inv(\pi)}}{\sum_{\pi \in \Sno}
    \phi^{-\inv(\pi)}}
  \end{align*}

  Next, we observe that we can similarly ignore inversions between two elements
  that are neither $1$ nor $n$. To see this, let $\inv_{1,n}(\pi)$ be the number
  of inversions involving at least one of $1$ and $n$. Then, if we fix the order
  and positions of $1$ and $n$, all possible permutations of the remaining
  elements $2$ through $n-1$ produce the same number of inversions
  $\inv_{1,n}(\pi)$. More formally, let $\pi_{(1)}$ and $\pi_{(n)}$ be the
  respective positions of elements $1$ and $n$. Then, this we have
  \begin{align*}
    \sum_{\pi \in \Son} \phi^{-\inv(\pi)}
    &= \sum_{k < \ell} \sum_{\pi : \pi_{(1)} = k, \pi_{(n)} = \ell}
    \phi^{-\inv(\pi)} \\
    &= \sum_{k < \ell} \sum_{\pi : \pi_{(1)} = k, \pi_{(n)} = \ell}
    \phi^{-\inv_{1,n}(\phi)} \cdot 
    \phi^{\inv_{1,n}(\phi)-\inv(\pi)} \\
    &= \sum_{k < \ell} \phi^{-(k-1) - (n-\ell)} \sum_{\pi : \pi_{(1)} = k,
    \pi_{(n)} = \ell} \phi^{\inv_{1,n}(\phi)-\inv(\pi)}
  \end{align*}
  As noted above, $\sum_{\pi : \pi_{(1)} = k, \pi_{(n)} = \ell}
  \phi^{\inv_{1,n}(\phi)-\inv(\pi)}$ does not depend on $k$ or $\ell$, since
  every permutation of the remaining elements yields the same number of
  inversions among them regardless of $k$ and $\ell$. A similar argument yields
  \begin{align*}
    \sum_{\pi \in \Sno} \phi^{-\inv(\pi)}
    &= \sum_{k > \ell} \phi^{-(k-1) - (n-\ell) + 1} \sum_{\pi : \pi_{(1)} = k,
    \pi_{(n)} = \ell} \phi^{\inv_{1,n}(\phi)-\inv(\pi)}
  \end{align*}
  Putting these together, we have
  \begin{align*}
    \frac{\Pr[1 \succ n]}{\Pr[n \succ 1]}
    &= \frac{\sum_{k < \ell} \phi^{-(k-1) - (n-\ell)} \sum_{\pi : \pi_{(1)} = k,
      \pi_{(n)} = \ell} \phi^{\inv_{1,n}(\phi)-\inv(\pi)}}{\sum_{k > \ell}
      \phi^{-(k-1) - (n-\ell)+1} \sum_{\pi : \pi_{(1)} = k, \pi_{(n)} = \ell}
    \phi^{\inv_{1,n}(\phi)-\inv(\pi)}} \\
    &= \frac{\sum_{k < \ell} \phi^{-(k-1) - (n-\ell)}}{\sum_{k > \ell}
    \phi^{-(k-1) - (n-\ell)+1}} \\
    &= \frac{\sum_{k < \ell} \phi^{-(k-1) - (n-\ell)}}{\sum_{k > \ell}
    \phi^{-(k-1) - (n-\ell)+1}} \cdot \frac{\phi^{n-1}}{\phi^{n-1}} \\
    &= \frac{\sum_{k < \ell} \phi^{\ell-k}}{\sum_{k > \ell}
    \phi^{\ell-k+1}}
  \end{align*}
  Note that each term in the numerator is strictly increasing in $\phi$, while
  each term in the denominator is weakly decreasing in $\phi$. As a result,
  $\frac{d}{d\phi} \frac{\Pr[1 \succ n]}{\Pr[n \succ 1]} > 0$, meaning for any $i
  < j$, $\frac{d}{d\phi} \Pr[i \succ j] > 0$.
\end{proof}

With this, we proceed inductively, showing that when $\phi_H \ge \phi_A$, each
firm rationally chooses to use $H$. For the first firm, by
Lemma~\ref{lem:monotone}, $H$ is more likely than $A$ to correctly order any
pair of candidates, meaning it produces higher expected utility. Similarly, for
any subsequent firm, conditioned on the remaining candidates, $H$ is still more
likely to correctly order any pair of remaining candidates, meaning it leads to
higher expected utility. A similar argument shows that all firms strictly prefer
$H$ when $\phi_H > \phi_A$.

\end{document}